\documentclass[submission,copyright,creativecommons]{eptcs}
\providecommand{\event}{QPL 2020} 
\usepackage{breakurl}             
\usepackage{underscore}           

\usepackage{amsfonts}
\usepackage{amsmath}
\usepackage{amssymb}
\usepackage{array}
\usepackage{amsthm}
\usepackage{mathtools}
\usepackage{mathrsfs}
\usepackage[svgnames]{xcolor}
\usepackage{tikz}
\usetikzlibrary{angles,quotes,shapes,shapes.misc,positioning}
\usepackage{stmaryrd}
\usepackage{comment}
\usepackage{tikz-cd}
\usepackage[utf8]{inputenc}
\usepackage{float}
\usepackage{thm-restate}
\usepackage{longtable}
\usepackage{csquotes}
\usepackage{placeins}

\title{Entanglement and Quaternions: The graphical calculus ZQ}
\author{Hector Miller-Bakewell
\institute{University of Oxford, UK}
\email{hector.miller-bakewell@cs.ox.ac.uk}
}
\def\titlerunning{The graphical calculus ZQ}
\def\authorrunning{H Miller-Bakewell}

\allowdisplaybreaks

\DeclareUnicodeCharacter{27E9}{\rangle}
\DeclareUnicodeCharacter{3C0}{\pi}

\newtheorem{theorem}{Theorem}[section]
\newtheorem{corollary}[theorem]{Corollary}
\newtheorem{proposition}[theorem]{Proposition}
\newtheorem{lemma}[theorem]{Lemma}

\theoremstyle{definition}
\newtheorem{definition}[theorem]{Definition}
\theoremstyle{remark}
\newtheorem{remark}[theorem]{Remark}
\newtheorem{example}[theorem]{Example}

\newcommand{\vc}[1]{
    \raisebox{-.5\height}{#1}
}

\newcommand{\tikzfig}[1]{\vc{\InputIfFileExists{./figures/\detokenize{#1}.tikz}{}{Missing file!}}}

\usetikzlibrary{arrows, decorations.markings, shapes.geometric, decorations.pathmorphing, decorations.pathreplacing, intersections, patterns,calc, backgrounds}
\tikzcdset{arrow style=tikz, diagrams={>=latex}}

\pgfdeclarelayer{bg}
\pgfdeclarelayer{mid}
\pgfdeclarelayer{fg}
\pgfdeclarelayer{edgelayer}
\pgfdeclarelayer{nodelayer}
\pgfsetlayers{bg,edgelayer,mid,nodelayer,main,fg}

\tikzset{
	0c/.style={circle, draw, fill, inner sep=1.5pt},
	1c/.style={->, thick, shorten <=2pt, shorten >=2pt},
	1cboth/.style={<->, thick, shorten <=2pt, shorten >=2pt},
	1clong/.style={->, thick},
	1cthin/.style={->, shorten <=4pt, shorten >=4pt},
	1cdot/.style={->, dashed, thick, shorten <=2pt, shorten >=2pt},
	1cinc/.style={right hook->, thick, shorten <=2pt, shorten >=2pt},
	1cincl/.style={left hook->, thick, shorten <=2pt, shorten >=2pt},
	follow/.style={->, >=stealth, very thick, shorten <=3pt, shorten >=3pt, color=magenta},
	2c/.style={double, thick, shorten <=6pt, shorten >=8pt, decoration={markings,mark=at position -6pt with {\arrow[scale=1.75]{>}}}, preaction={decorate}},
	2cdot/.style={double, dashed, thick, shorten <=10pt, shorten >=10pt, decoration={markings,mark=at position -8pt with {\arrow[scale=1.75]{>}}}, preaction={decorate}},
	3c1/.style={thick, double, double distance=3pt, shorten <=9pt, shorten >=11pt},
    	3c2/.style={thick, shorten <=9pt, shorten >=10pt},
	3c3/.style={shorten <=9pt, shorten >=10pt, decoration={markings,mark=at position -8pt with {\arrow[scale=3]{>}}},preaction={decorate}},
	4c1/.style={thick, double, double distance=4pt, shorten <=1pt, shorten >=2.75pt},
	4c2/.style={thick, double, double distance=1pt, shorten <=1pt, shorten >=1.25pt, decoration={markings,mark=at position -.05pt with {\arrow[scale=3,ultra thin]{>}}},preaction={decorate}},
	edge/.style={line width=.8pt, color=black},
	edgedot/.style={densely dotted, line width=.8pt, color=black},
	edgethdot/.style={densely dotted, line width=.4pt, color=gray},
	edgeth/.style={line width=.4pt, color=gray!60},
	edgethin/.style={line width=.8pt, color=gray!60},
	edgedotdark/.style={densely dotted, line width=.8pt, color=gray!80},
	dot/.style={circle, draw=black, line width=.8pt, fill=white, inner sep=1.7pt},
	dotth/.style={circle, draw=gray!60, fill=gray!60, inner sep=1.5pt},
	dotwh/.style={circle, draw=gray!60, line width=.4pt, fill=white, inner sep=1.7pt},
	dotwhite/.style={circle, draw=black, line width=.8pt, fill=white, inner sep=1.8pt},
	dotdark/.style={circle, draw, fill=black, inner sep=1.5pt},
	dotgrey/.style={circle, draw=black, line width=.8pt, fill=gray!60, inner sep=1.8pt},
	trian/.style={regular polygon,regular polygon sides=3,shape border rotate=0,fill=white, line width=.8pt, draw=black, inner sep=1.8pt},
	trianh/.style={regular polygon,regular polygon sides=3,shape border rotate=0,fill=white, draw=gray!60, line width=.4pt, inner sep=1.8pt},
	trib/.style={regular polygon,regular polygon sides=3,shape border rotate=0,fill=black, draw, inner sep=1.5pt},
	tribh/.style={regular polygon,regular polygon sides=3,shape border rotate=0,fill=gray!60, draw=gray!60, inner sep=1.5pt},
	tribco/.style={regular polygon,regular polygon sides=3,shape border rotate=180,fill=black, draw, inner sep=1.5pt},
	cover/.style={circle, draw=gray!10, fill=gray!10, inner sep=3.5pt},
	coverc/.style={circle, draw=gray!80, line width=.4pt, inner sep=3pt},
	coverch/.style={circle, draw=gray!30, line width=.4pt, inner sep=3pt},
	coverb/.style={circle, draw=gray!80, line width=.4pt, fill=gray!10, inner sep=3.5pt},
	every node/.style={font={\scriptsize}},
	every matrix/.append style={nodes={font=\normalsize}}
}
\def\roundingAmount{0.5em}
\def\paddingAmount{1em}
\def\sepAmount{0.2em}
\def\minHeight{1em}
\def\minWidth{1em}
\def\tinyBoxSize{7pt}
\def\tinyCircleSize{7pt}
\definecolor{zx_red}{RGB}{232, 165, 165}
\definecolor{zx_green}{RGB}{216, 248, 216}
\definecolor{nice_green}{RGB}{116, 148, 116}
\def\hyellow{yellow!30}
\tikzstyle{spider}=[rectangle,rounded corners=\roundingAmount,fill=gray!1,draw=Black,
  line width=0.2 pt,inner sep=\sepAmount,minimum width=\paddingAmount,minimum height=\paddingAmount]
\tikzstyle{rect}=[rectangle,fill=gray!1,draw=Black,inner sep=\sepAmount,minimum width=\minWidth,minimum height=\minHeight,
  line width=0.2 pt]
\tikzstyle{box}=[rect,fill=white, line width=0.4 pt]
\tikzstyle{trap}=[trapezium,trapezium angle=80,fill=gray!1,draw=Black,inner sep=\sepAmount,minimum width=\minWidth,minimum height=\minHeight,
  line width=0.2 pt]
\tikzstyle{triangle}=[regular polygon,regular polygon sides=3,draw=Black,inner sep=\sepAmount,minimum width=\minWidth,minimum height=\minHeight,
  line width=0.2 pt]
\tikzstyle{trianglet}=[triangle,shape border rotate=180]
\tikzstyle{ZH}=[rect]
\tikzstyle{wide ZH}=[rect,minimum width=2.5em]
\tikzstyle{smallH}=[rect, minimum width=\tinyBoxSize, minimum height =\tinyBoxSize]
\tikzstyle{smallZH}=[smallH]
\tikzstyle{yh}=[rect, minimum width=\tinyBoxSize, minimum height =\tinyBoxSize,fill=\hyellow]
\tikzstyle{smallCircle}=[circle,draw=Black, minimum width=\tinyCircleSize, line width=0.2 pt, inner sep=0em]
\tikzstyle{smallZ}=[smallCircle,fill=gray!1]
\tikzstyle{Z}=[smallZ]
\tikzstyle{wide point}=[fill=white,draw,shape=isosceles triangle,shape border rotate=-90,isosceles triangle stretches=true,inner sep=0pt,minimum width=2cm,minimum height=6.12mm,yshift=-0.0mm]
\tikzstyle{bbindex}=[font={\color{blue}\footnotesize}]
\tikzstyle{medium gray box}=[rect, minimum width=4em,fill=gray!30]
\tikzstyle{semilarge box}=[rect, minimum width=4em]
\tikzstyle{wide box}=[rect, minimum width=6em]
\tikzstyle{small gray box}=[rect, fill=gray!30]
\tikzstyle{H}=[rect,fill=\hyellow]
\tikzstyle{h}=[smallH,fill=\hyellow]
\tikzstyle{gn}=[spider,fill=zx_green]
\tikzstyle{rn}=[spider,fill=zx_red]
\tikzstyle{white}=[spider,fill=gray!1]
\tikzstyle{smallWhite}=[smallZ]
\tikzstyle{smallwhite}=[smallZ]
\tikzstyle{grey}=[spider,fill=gray!30]
\tikzstyle{smallGrey}=[smallCircle,fill=gray!30]
\tikzstyle{smallgrey}=[smallGrey]
\tikzstyle{black}=[spider,fill=gray!70]
\tikzstyle{small black}=[smallCircle,fill=gray!70]
\tikzstyle{smallBlack}=[small black]
\tikzstyle{smallblack}=[small black]
\tikzstyle{b}=[black]
\tikzstyle{w}=[white]
\tikzstyle{qn}=[trap]
\tikzstyle{qnt}=[trap,shape border rotate=180]
\tikzstyle{string}=[circle,fill=gray,draw=gray,inner sep=1pt]
\tikzstyle{net}=[rectangle,draw=White,minimum width=1.5em,minimum height=1.5em,fill=white]
\tikzstyle{none}=[inner sep=0pt]
\tikzstyle{crossing}=[circle,minimum width=\tinyBoxSize,draw=black,{path picture={ 
  \draw[black]
(path picture bounding box.south east) -- (path picture bounding box.north west) (path picture bounding box.south west) -- (path picture bounding box.north east);
}}]
\tikzstyle{bbox}=[rectangle,draw=blue!60!white,minimum width=1em,minimum height=1em]
\tikzstyle{polynomial}=[qn]
\tikzstyle{polynomialT}=[polynomial,shape border rotate=180]
\tikzstyle{poly}=[polynomial]
\tikzstyle{polyT}=[polynomialT]

\tikzstyle{arrow}=[->,draw=black]
\tikzstyle{hadamard edge}=[-, dashed, dash pattern=on 2pt off 1pt, thick, draw={rgb,255: red,68; green,136; blue,255}]
\tikzstyle{light-arrow}=[->,draw=gray]
\tikzstyle{blue}=[draw=blue!30!white]

\tikzstyle{CNOTcross}=[smallCircle, path picture={ 
\draw[thick,black](path picture bounding box.north) -- (path picture bounding box.south) (path picture bounding box.west) -- (path picture bounding box.east);
}]
\tikzstyle{CNOTblack}=[smallCircle, fill=black]

\pgfdeclarearrow{
  name=reduce,
  parameters={\the\pgfarrowlength},  
  setup code={
   \pgfarrowssettipend{0pt}
   \pgfarrowssetlineend{-\pgfarrowlength}
   \pgfarrowlinewidth=\pgflinewidth
   \pgfarrowssavethe\pgfarrowlength
  },
  drawing code={
   \pgfpathmoveto{\pgfpoint{0\pgfarrowlength}{0\pgfarrowlength}}
   \pgfpathlineto{\pgfpoint{-1\pgfarrowlength}{0.5\pgfarrowlength}}
   \pgfpathlineto{\pgfpoint{-1\pgfarrowlength}{-0.5\pgfarrowlength}}
   \pgfpathlineto{\pgfpoint{0\pgfarrowlength}{0\pgfarrowlength}}
   \pgfusepathqstroke
  },
  defaults = { length = 3pt }
}

\def\CNOT{\ensuremath{\text{CNOT}}}
\def\CZ{\ensuremath{\text{CZ}}}

\def\bbQ{\mathbb{Q}}

\def\id{\text{id}}

\def\bbC{\mathbb{C}}
\def\bbH{\mathbb{H}}
\def\bbN{\mathbb{N}}

\def\bbR{\mathbb{R}}

\def\Hilbert{\ensuremath{\mathbb{H}}}
\def\F2{{\mathbb{F}_2}}

\def\ZX{\text{ZX}}

\def\entails{\vdash}

\def\semantic{\vDash}

\newcommand{\set}[1]{\left\{#1\right\}}

\newcommand{\abs}[1]{|#1|}
\newcommand{\comp}{\circ}

\newcommand{\interpret}[1]{\left\llbracket \; #1 \; \right\rrbracket}
\def\tensor{\ensuremath\otimes}

\newcommand\bra[1]{\langle\, #1\,|}
\newcommand\ket[1]{|\, #1 \, \rangle}
\newcommand\ketbra[2]{\ket{#1} \! \bra{#2}}

\newcommand\restr[2]{{
  \left.\kern-\nulldelimiterspace 
  #1 
  \vphantom{\big|} 
  \right|_{#2}
  }}
\renewcommand*{\arraystretch}{1.2}

\newcommand\by[1]{\;
    \raisebox{-.5\height}{$
        \stackrel{=}{
            \mbox{ \scriptsize $#1$}
        }$
    }\;
}

\def\w{\omega}

\newcommand\ring{\text{\textsc{ring}}}

\def\monprop{\text{\textsc{monoid-prop}}}

\def\ie{\text{i.e.}\quad}

\def\ZX{\text{ZX}}

\newcommand{\QRingZH}[1][]{
	\ensuremath{\ring_{ZH}}
}
\newcommand\half[1][1]{\frac{#1}{2}}
\newcommand\piby[1]{\pi/#1}

\def\SO3{\ensuremath{\text{SO3}(\bbR)}}
\def\SU2{\ensuremath{\text{SU2}(\bbC)}}
\def\UQuat{\ensuremath{\hat{Q}}}
\def\ZWC{\ensuremath{\text{ZW}_\bbC}}

\def\contradiction{\raisebox{-0.25\height}{
	\begin{tikzpicture}[scale=0.2]
	\begin{pgfonlayer}{nodelayer}
		\node [style=none] (0) at (-0.5, 1) {};
		\node [style=none] (1) at (-1, 0.5) {};
		\node [style=none] (2) at (-1, -0.5) {};
		\node [style=none] (3) at (-0.5, -1) {};
		\node [style=none] (4) at (0.5, -1) {};
		\node [style=none] (5) at (1, -0.5) {};
		\node [style=none] (6) at (1, 0.5) {};
		\node [style=none] (7) at (0.5, 1) {};
	\end{pgfonlayer}
	\begin{pgfonlayer}{edgelayer}
		\draw (0.center) to (5.center);
		\draw (1.center) to (4.center);
		\draw (2.center) to (7.center);
		\draw (6.center) to (3.center);
	\end{pgfonlayer}
\end{tikzpicture}}}

\def\rowgap{2em}

\DeclareUnicodeCharacter{27E9}{\rangle}
\DeclareUnicodeCharacter{3C0}{\pi}

\newcommand{\node}[2]{
	\raisebox{-0.5\height}{
		\begin{tikzpicture}[scale=0.3]
			\begin{pgfonlayer}{nodelayer}
				\node [style=#1] (1)   at (0, 0)   {$#2$};
				\node [style=none] (2)   at (0,2)   {};
				\node [style=none] (3)   at (0,-2)   {};
			\end{pgfonlayer}
			\begin{pgfonlayer}{edgelayer}
				\draw [-] (1.center) to (2.center);
				\draw [-] (1.center) to (3.center);
			\end{pgfonlayer}
		\end{tikzpicture}
	}
}

\newcommand{\galpharpi}[1]{
	\vc{
		\begin{tikzpicture}[scale=0.3]
			\begin{pgfonlayer}{nodelayer}
				\node [style=gn] (1)   at (0, 1.2)   {$#1$};
				\node [style=gn] (2)   at (0, -1.2)   {$\pi$};
			\end{pgfonlayer}
			\begin{pgfonlayer}{edgelayer}
				\draw [hadamard edge] (1.center) to (2.center);
			\end{pgfonlayer}
		\end{tikzpicture}
	}
}

\def\tripleblobs{
	\vc{
		\begin{tikzpicture}[scale=0.3]
			\begin{pgfonlayer}{nodelayer}
				\node [style=gn] (1)   at (0, 1.2)   {};
				\node [style=rn] (2)   at (0, -1.2)   {};
			\end{pgfonlayer}
			\begin{pgfonlayer}{edgelayer}
				\draw [-] (1.center) to (2.center);
				\draw [-,bend left=35] (1.center) to (2.center);
				\draw [-,bend right=35] (1.center) to (2.center);
			\end{pgfonlayer}
		\end{tikzpicture}
	}
}
\def\halfblobs{\tripleblobs \tripleblobs}

\newcommand{\rg}[2]{
	\vc{
		\begin{tikzpicture}[scale=0.3]
			\begin{pgfonlayer}{nodelayer}
				\node [style=gn] (1)   at (0, 1.2)   {$#1$};
				\node [style=gn] (2)   at (0, -1.2)   {$#2$};
			\end{pgfonlayer}
			\begin{pgfonlayer}{edgelayer}
				\draw [hadamard edge] (1.center) to (2.center);
			\end{pgfonlayer}
		\end{tikzpicture}
	}
}

\newcommand{\state}[2]{
  \raisebox{0.2em}{
	\vc{
		\begin{tikzpicture}[scale=0.3]
			\begin{pgfonlayer}{nodelayer}
				\node [style=#1] (1)   at (0, 0)   {$#2$};
				\node [style=none] (2)   at (0,2)   {};
			\end{pgfonlayer}
			\begin{pgfonlayer}{edgelayer}
				\draw [-] (1.center) to (2.center);
			\end{pgfonlayer}
		\end{tikzpicture}
	}
}
}

\newcommand{\hedge}{
	\raisebox{-0.5\height}{
		\begin{tikzpicture}[scale=0.3]
			\begin{pgfonlayer}{nodelayer}
				\node [style=none] (2)   at (0, 2)   {};
				\node [style=none] (3)   at (0, -2)   {};
			\end{pgfonlayer}
			\begin{pgfonlayer}{edgelayer}
				\draw [hadamard edge] (3.center) to (2.center);
			\end{pgfonlayer}
		\end{tikzpicture}
	}
}
\newcommand{\spider}[2]{
	\raisebox{-0.5\height}{
		\begin{tikzpicture}[scale=0.3]
			\begin{pgfonlayer}{nodelayer}
				\node [style=#1] (1)   at (0, 0)   {$#2$};
				\node [style=none] (2)   at (-1,2)   {};
				\node [style=none] ()   at (0,2)   {\raisebox{0.5\height}{$\dots$}};
				\node [style=none] (3)   at (1,2)   {};
				\node [style=none] (4)   at (-1, -2)   {};
				\node [style=none] ()   at (0,-2)   {\raisebox{0.5\height}{$\dots$}};
				\node [style=none] (5)   at (1, -2)   {};
			\end{pgfonlayer}
			\begin{pgfonlayer}{edgelayer}
				\draw [-] (1.center) to (2.center);
				\draw [-] (1.center) to (3.center);
				\draw [-] (1.center) to (4.center);
				\draw [-] (1.center) to (5.center);
			\end{pgfonlayer}
		\end{tikzpicture}
	}
}

\newcommand{\binary}[2][ZH]{
  \raisebox{0.2em}{
	\vc{
		\begin{tikzpicture}[scale=0.3]
			\begin{pgfonlayer}{nodelayer}
				\node [style=#1] (1)   at (0, 0)   {$#2$};
				\node [style=none] (2)   at (0, 2)   {};
				\node [style=none] (4)   at (-1, -2)   {};
				\node [style=none] (5)   at (1, -2)   {};
			\end{pgfonlayer}
			\begin{pgfonlayer}{edgelayer}
				\draw [-] (1.center) to (2.center);
				\draw [-] (1.center) to (4.center);
				\draw [-] (1.center) to (5.center);
			\end{pgfonlayer}
		\end{tikzpicture}
	}
}
}

\newcommand{\spidermn}[4]{
	\vc{
		\begin{tikzpicture}[scale=0.3]
			\begin{pgfonlayer}{nodelayer}
				\node [style=#1] (1)   at (0, 0)   {$#2$};
				\node [style=none] (2)   at (-1, 2)   {};
				\node [style=none] ()   at (0, 2.5)   {$\overset{#4}{\dots}$};
				\node [style=none] (3)   at (1, 2)   {};
				\node [style=none] (4)   at (-1, -2)   {};
				\node [style=none] ()   at (0, -2.5)   {$\underset{#3}{\dots}$};
				\node [style=none] (5)   at (1, -2)   {};
			\end{pgfonlayer}
			\begin{pgfonlayer}{edgelayer}
				\draw [-] (1.center) to (2.center);
				\draw [-] (1.center) to (3.center);
				\draw [-] (1.center) to (4.center);
				\draw [-] (1.center) to (5.center);
			\end{pgfonlayer}
		\end{tikzpicture}
	}
}

\def\dcup{\raisebox{0.1em}{\tikzfig{wire/cup}}}
\def\dcap{\raisebox{0.3em}{\tikzfig{wire/cupt}}}

\begin{document}
\maketitle

\begin{abstract}
	Graphical calculi are vital tools for representing and reasoning about quantum circuits and processes.
	Some are not only graphically intuitive but also logically complete.
	The best known of these is the ZX-calculus,
	which is an industry candidate for an Intermediate Representation;
	a language that sits between the algorithm designer’s intent and the quantum hardware’s gate instructions.
	The ZX calculus, built from generalised Z and X rotations, has difficulty reasoning about arbitrary rotations.
	This contrasts with the cross-hardware compiler TriQ which uses these arbitrary rotations to exploit hardware efficiencies.
	In this paper we introduce the graphical calculus ZQ,
	which uses quaternions to represent these arbitrary rotations, similar to TriQ,
	and the phase-free Z spider to represent entanglement, similar to ZX.
	We show that this calculus is sound and complete for qubit quantum computing,
	while also showing that a fully spider-based representation would have been impossible.
	This new calculus extends the zoo of qubit graphical calculi, each with different strengths,
	and we hope it will provide a common language for the optimisation procedures of both ZX and TriQ.
\end{abstract}

\section{Introduction}

The purpose of this paper is to introduce a new graphical calculus,
called ZQ,
similar to the already established graphical calculi of ZX, ZW and ZH.
These calculi are \emph{universal}, \emph{sound}, and \emph{complete}
as representations of qubit quantum computing circuits:
Any circuit can be represented as a diagram in any of these calculi,
and two circuits perform the same operation on qubits
if and only if the rules of the calculus show an equality
between the corresponding diagrams.
The ZH calculus \cite{ZH} and the ZW calculus \cite{AmarThesis, ZW}
are based on the algebraic structure of rings
(for qubit quantum computing we explicitly mean the calculus \ZWC).
The ZX calculus and, as we shall see, the ZQ calculus
are instead based on group structures.
This similarity in algebraic structure will
be used to find a translation between the two calculi,
providing us with a method to show the universality and completeness of ZQ,
but also highlights an important difference:
ZQ is built on a non-commutative group,
but we will show that the fundamental building blocks of ZX are restricted to commutative groups.

The ZX calculus is built from the Z and X classical structures of quantum computing, and was introduced in Ref.~\cite{Coecke08}.
Even in that earliest paper the Z `phase shift' is illustrated as a rotation of the Bloch Sphere \cite[\S 4]{Coecke08}.
By the time of Ref.~\cite{Backens16Thesis}, eight years later,
language had changed to that of Z `rotations' or `angles' \cite[Lemma~3.1.7]{Backens16Thesis},
and explicit use is made of the Euler Angle Decomposition result;
that any rotation in \SO3 can be broken down into rotations about the Z then X then Z axes.
The idea behind the calculus ZQ is to represent not just the Z and X rotations of the Bloch Sphere,
but represent arbitrary rotations via unit-length quaternions.
ZX is built not just from rotations but also from \emph{spiders}:
Rotations are viewed as acting on individual qubits,
but spiders link multiple qubits, expressing entanglement.
The observation of Ref.~\cite{Coecke08} is that the
structures of spiders and rotations can be merged
into a single diagrammatic entity,
where each spider (see Figure~\ref{figZXSpiders}) is given a \emph{colour}, of either green indicating Z or red indicating X,
and a \emph{phase}, indicating the angle of rotation.

\begin{figure}
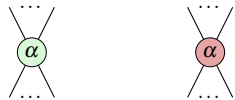

	\begin{align*}
		\spider{gn}{\alpha} \qquad & \qquad \spider{rn}{\alpha}
	\end{align*}
	\caption{\label{figZXSpiders}The spiders of ZX can have any number of inputs and outputs,
		have a colour of red or green, and are labelled by an angle.
		If the angle is 0 it is often omitted.
		The colours used in this paper have been chosen such that Z (green) should appear lighter than X (red), even when viewed in greyscale \cite{ZXAccessibility}.}
\end{figure}

The Bloch Sphere, which we cover in more detail in \S\ref{secBlochSphere}, is not a perfect analogy \cite{Wharton2015}.
Although it provides us with useful intuition and a way to consider a single qubit in real Euclidean space,
its group of rotations, \SO3, is a subgroup of the group of special unitary evolutions, \SU2,
which the standard circuit model of quantum computing actually uses \cite{NielsonChuang}.
The group \SU2 itself is isomorphic to the group of unit-length quaternions,
and so we shall use these quaternions as a replacement for \SO3's rotations,
giving us the `Q' in ZQ.
This use of quaternions to represent rotations is not new to quantum computing \cite{Wharton2015},
nor other domains such as engineering or computer graphics \cite{Shoemake},
but has recently surfaced as a useful component of Intermediate Representations for quantum circuits.
Intermediate Representations sit between the user's specification of an algorithm
and the actual implementation on a specific piece of hardware.
The system TriQ \cite{1905.11349}
provides such an Intermediate Representation, targeting existing quantum computers run by IBM, Rigetti, and the University of Maryland.
The authors claim a speed-up in execution of their benchmarks on the seven quantum computers considered,
in part because of TriQ's use of quaternions in the optimisation process \cite[\S4]{TRIQ}:
Any sequence of single qubit gates can be combined into just one quaternion,
then decomposed into the most efficient sequence of gates for the target hardware architecture.

Our aims in making ZQ are the following:
\begin{itemize}
	\item Construct a graphical calculus that succinctly expresses all single qubit operations
	\item Provide a complete graphical calculus that can express the Intermediate Representation of TriQ
	\item Construct a qubit graphical calculus whose phases form a non-commutative group
\end{itemize}

Before we give the definition of ZQ we first give a brief overview of
the Bloch Sphere, the groups \SU2 and \SO3, and unit quaternions in \S\ref{secBlochSphere}.
In \S\ref{secZX} we describe the graphical calculus ZX,
in \S\ref{secSpiders} we show why the spiders of ZX are incompatible with non-commutative groups, and then
in \S\ref{secZQ} we introduce the graphical calculus ZQ and demonstrate its universality, soundness and completeness.

\section{Rotations, Quaternions and TriQ}
\label{secBlochSphere}

The book `Quantum computation and quantum information' \cite{NielsonChuang}
defines a qubit as a unit vector in $\Hilbert := \bbC^2$,
but notes that if two qubits differ by a unit complex scalar then
they result in the same experimental observations.
The paper `Unit Quaternions and the Bloch Sphere' \cite{Wharton2015}
instead uses the term \emph{spinor} for a unit vector in $\Hilbert$,
often represented by a quaternion, with rotations also being represented as quaternions,
and the term qubit to mean a point in the quotient space $\Hilbert / v \sim e^{i \alpha} v$.
We are taking care to highlight this difference
because this paper will take qubits as unit vectors in $\bbC^2 = \Hilbert$
as in Ref.~\cite{NielsonChuang}
but will be using quaternions to represent rotations
in a manner related to Ref.~\cite{Wharton2015}.
Our reason for this is to make best use of
the $\bbC$-tensor product of $(\bbC^2)^{\tensor n}$,
allowing us to use the paradigm of Categorical Quantum Mechanics \cite{CQM},
but to also provide a tidy representation of the group of rotations.

\begin{definition}[Qubits and the Bloch Sphere]  \label{defQubitAndBlochSphere} \cite{NielsonChuang}
	A qubit is a unit vector in $\Hilbert := \bbC^2$.
	Two qubits $v$ and $v'$ are considered experimentally indistinguishable
	if $v' = e^{i\alpha} v$, defining the equivalence relation $v' \sim v$.
	Any qubit $v$ is equivalent via this relation to a vector defined
	just using two angles, $\theta$ and $\phi$.
	\begin{align}
		\ket{0} & := \begin{pmatrix}
			1 \\ 0
		\end{pmatrix} &
		\ket{1} & := \begin{pmatrix}
			0 \\ 1
		\end{pmatrix} &
		v& \sim e^{-i \phi / 2} \cos \half[\theta] \ket{0} + e^{i \phi / 2} \sin \half[\theta] \ket{1}
	\end{align}
	The space of qubits quotient the relation $\sim$ is called the Bloch Sphere,
	with a `qubit up to global phase' given by the spherical coordinates
	$(\theta, \phi)$.
\end{definition}

\begin{definition}[Rotations of the Bloch Sphere]  \label{defQubits} \cite[\S{2}]{Wharton2015}
	The Bloch Sphere is the familiar 2-sphere in 3-dimensional real space.
	Accordingly its group of rotations is \SO3.
\end{definition}

This presentation of rotations of the Bloch Sphere corresponds
to the naming of the Pauli $X$, $Y$, and $Z$ matrices
as those that fix the $x$, $y$, and $z$ axes.
This correspondence, however, is imperfect:
The Bloch Sphere has already discarded the global phase,
but the Pauli matrices act on qubits.
Rather than continue to use the language of 3D rotations
we shall instead be using unit quaternions
(via group isomorphism with \SU2) to label our fundamental, single-qubit evolutions.
Quaternions are a four-dimensional real algebra,
in the same way that the complex numbers are a two-dimension real algebra.

\begin{definition}[Quaternions]  \label{defQuaternions} \cite[p12]{Hazewinkel2004algebras}
	The quaternions, invented by Hamilton in 1843, are a non-commutative, four-dimensional, real algebra:
	\begin{align}
		\bbR + i\bbR + j \bbR + k \bbR \qquad \qquad
		i^2 = j^2 = k^2 = ijk = -1
	\end{align}
	For ZQ we are only interested in unit-length quaternions, forming the group \UQuat\ under multiplication.
	The group \UQuat\ is isomorphic with \SU2, via the isomorphism:
	\begin{align}
		\phi: \UQuat & \to \SU2     \label{eqnUQuatSU2} &
		q_w + iq_x + jq_y + kq_z & \mapsto \begin{pmatrix}
			q_w - iq_z & -q_y + iq_x \\
			-q_y-iq_x  & q_w + iq_z
		\end{pmatrix}
	\end{align}
	The proof that this is an isomorphism is given as Proposition~\ref{propZQPhi}.
\end{definition}

At first glance \SO3\ and \UQuat\ may seem to be unrelated mathematical entities, but
there is another way to represent unit-length quaternions, and that is by an angle and a unit vector.
It is important to note that this is not the same thing as `an angle rotation along a unit vector':
The angle-vector pair $(\alpha, \hat v)$ and the angle-vector pair $(-\alpha, -\hat v)$ are different as pairs,
but would constitute the same rotation in \SO3. This, in fact, describes the relationship between \UQuat\ and \SO3.

\begin{definition}[Relating unit quaternions to \SO3]  \label{defUQUatSO3}
	There is a canonical homomorphism from \UQuat\ to \SO3, given by
	\begin{align}
		\psi         : \UQuat  &\to \SO3 & (\alpha, v) &:= \cos \half[\alpha] + \sin\half[\alpha](iv_x + jv_y + kv_z)       \\
		(\alpha, v) &\mapsto \text{ rotation by angle $\alpha$ along vector $v$} &  \ker \psi  &= \set{1, -1}    
	\end{align}
\end{definition}

\begin{remark} \label{remPauliXYZ}
	The axes $x$, $y$ and $z$ relate these quaternions as rotations (via $\phi$) to the Pauli matrices $X$, $Y$, and $Z$.
	Our presentation introduces a scale factor of $\pm i$, similar to that in Ref.~\cite[Table~1]{Wharton2015}.

	\begin{align}
		\phi((\pi, x))                         & =
		-i
		\begin{pmatrix}
			0 & 1 \\
			1 & 0
		\end{pmatrix} = -iX
		                                       &
		\phi((\pi, y))                         & =
		-i
		\begin{pmatrix}
			0  & i \\
			-i & 0
		\end{pmatrix} = iY
		\\
		\phi((\pi, z))                         & =
		-i
		\begin{pmatrix}
			1 & 0  \\
			0 & -1
		\end{pmatrix} = -iZ
		                                       &
		\phi((\pi, \frac{x + z}{\sqrt{2}})) & =
		\frac{-i}{\sqrt{2}}
		\begin{pmatrix}
			1 & 1  \\
			1 & -1
		\end{pmatrix} = -iH
	\end{align}
\end{remark}

Since unit quaternions can represent the fundamental
single qubit operations (with $\psi$ linking composition of operations
to multiplication of quaternions)
it can be simpler to just use the quaternion representation,
as in the example of the cross-hardware compiler TriQ:

\begin{example}[Quaternions in TriQ] \label{exaQuaternionsInTriQ}
	The compiler TriQ uses quaternions as part of its optimisation process.
	\begin{displayquote}[Full-Stack, Real-System Quantum Computer Studies:
			Architectural Comparisons and Design Insights \cite{TRIQ}]
		Since 1Q operations are rotations, each
		1Q gate in the [Intermediate Representation] can be expressed using a unit rotation quaternion
		which is a canonical representation using a 4D complex number.
		TriQ composes rotation operations by multiplying the corresponding
		quaternions and creates a single arbitrary rotation. This rotation is
		expressed in terms of the input gate set. Furthermore, on all three
		vendors, Z-axis rotations are special operations that are implemented
		in classical hardware and are therefore error-free. TriQ expresses
		the multiplied quaternion as a series of two Z-axis rotations and one
		rotation along either X or Y axis, thereby maximizing the
		number of error-free operations.
	\end{displayquote}
\end{example}

We shall explicitly construct this decomposition of a quaternion
into a Z-X-Z rotation in Proposition~\ref{propZQQDecomp}
when we explore how to translate from ZQ to ZX.
With these notions of rotations and quaternions established we turn to the ZX-calculus.

\section{The ZX-calculus} \label{secZX}

The ZX-calculus is a graphical calculus
similar to the usual quantum circuit notation of e.g. Ref.~\cite{NielsonChuang}.
We provide here only a brief introduction, for more see Ref.~\cite{PQP}.
ZX-diagrams are built from red (X) and green (Z) spiders, as shown in Figure~\ref{figZXSpiders},
joined by wires.
These spiders can have any number of inputs or outputs,
and they, along with the wires, form the building blocks of the diagrams.
Two diagrams can be placed side by side (horizontal composition, $\tensor$)
or the outputs of one are plugged into the inputs of another above
(vertical composition, $\comp$).
Note that these diagrams are read bottom-to-top,
rather than left-to-right,
but this is purely a matter of convention.

These spiders and wires represent linear maps,
with the notation $\interpret{D}$ indicating the linear map associated with the diagram D (see Figure~\ref{figZXInterpretation}).
Indeed the calculus is \emph{universal}
in that any linear map $M : \Hilbert^{\tensor m} \to \Hilbert^{\tensor n}$
can be represented as a ZX-diagram.
The calculus also comes with a set of rules,
and these rules are \emph{complete},
meaning that if two diagrams represent the same linear map
then one can be transformed to the other by the rules.
In fact there are several fragments of ZX,
each of which can be seen as a restriction on the available Z and X rotations,
and each of which has a complete ruleset:
Stabilizer ZX \cite{Backens13},
Clifford+T ZX \cite{JPV17},
various finite subgroups beyond Clifford+T ZX \cite{BeyondCliffordT},
and the Universal ZX \cite{UniversalComplete}.
We shall be looking at just the last of these in this paper,
and the ruleset we shall be considering is given in Figure~\ref{figZXRules}.

Note that the rule (EU') of Figure~\ref{figZXRules}
includes a long description of the calculation of $\beta_1$, $\beta_2$, $\beta_3$
and $\gamma$.
This is known as a \emph{side condition},
and the complexity of the condition stems from the equivalence of different Euler Angle Decompositions.
There are other complete rulesets for Universal ZX \cite{UniversalComplete}  \cite{BeyondCliffordT},
but each has a side condition requiring the calculation of moduli and arguments of complex numbers.
When we reach the definition of ZQ we will see that there is no such side condition
related to Euler Angle Decompositions\footnote{ZQ's side condition in the (Y) rule relates to the transpose.}, because
it is inherent in the group action of \UQuat\ and the rule (Q).

\begin{figure}
	\begin{align}
		\interpret{\spider{gn}{\alpha}} & = \ketbra{0...0}{0...0} +		e^{i \alpha} \ketbra{1...1}{1...1}  &
		\interpret{\ \node{none}{}\ }          & = \begin{pmatrix}		1 & 0 \\ 0 & 1		\end{pmatrix} \\
		\interpret{\spider{rn}{\alpha}} & =  \ketbra{+...+}{+...+} +		e^{i \alpha} \ketbra{-...-}{-...-} &
		\interpret{\dcap}                  & = \begin{pmatrix}		1 & 0 & 0 & 1		\end{pmatrix}\\
		\interpret{\tikzfig{wire/swap}} & =  \begin{pmatrix}
			1 & 0 & 0 & 0 \\
			0 & 0 & 1 & 0 \\
			0 & 1 & 0 & 0 \\
			0 & 0 & 0 & 1
		\end{pmatrix}                                &
		\interpret{\dcup}                  & = \begin{pmatrix}		1 \\ 0 \\ 0 \\ 1		\end{pmatrix}
	\end{align}
	\caption{The interpretation of the elements of ZX-calculus diagrams \label{figZXInterpretation}}
\end{figure}

\begin{figure}
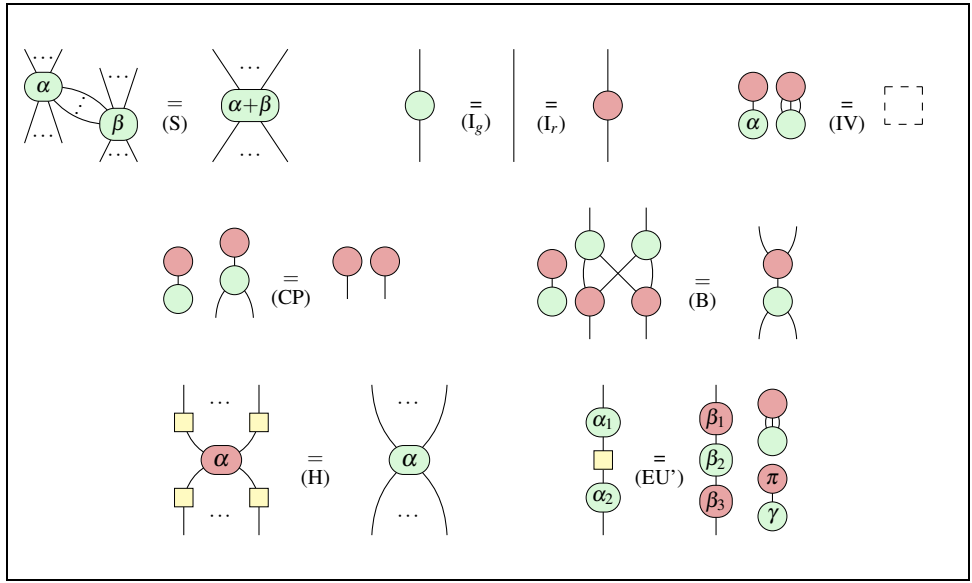

	\centering
	\scalebox{1}{\begin{tabular}{|c|}
			\hline                                                                                                                         \\
			\tikzfig{ZX/Vilmart/spider-1}$\qquad\quad~$\tikzfig{ZX/Vilmart/s2-green-red}$\qquad\quad~$\tikzfig{ZX/Vilmart/inverse-param}~~ \\\\
			\tikzfig{ZX/Vilmart/b1s}$\qquad\qquad$\tikzfig{ZX/Vilmart/b2s}                                                                 \\\\
			\tikzfig{ZX/Vilmart/h2}$\qquad\qquad$\tikzfig{ZX/Vilmart/Euler-Had}                                                            \\\\
			\hline
		\end{tabular}}
	\caption[]{Set of rules ZX for the ZX-Calculus with scalars from Ref.~\cite{VilmartZX}. The right-hand side of (IV) is an empty diagram.
		(...) denote zero or more wires, while (\protect\rotatebox{45}{\raisebox{-0.4em}{$\cdots$}}) denote one or more wires. In rule (EU'), $\beta_1,\beta_2,\beta_3$ and $\gamma$ can be determined as follows:
		$x^+:=\frac{\alpha_1+\alpha_2}{2}$, $x^-:=x^+-\alpha_2$, $z := -\sin{x^+}+i\cos{x^-}$ and $
			z' := \cos{x^+}-i\sin{x^-}$, then
		$\beta_1 = \arg z + \arg z',
			\beta_2 = 2\arg\left(i+\left|\frac{z}{z'}\right|\right),
			\beta_3 = \arg z - \arg z',
			\gamma = x^+-\arg(z)+\frac{\pi-\beta_2}{2}$
		where by convention $\arg(0):=0$ and $z'=0\implies \beta_2=0$.
	}
	\label{figZXRules}
\end{figure}

\begin{example}[Quantum circuits are ZX diagrams] \label{exaQunatumCircuitsAsZX}

	Quantum circuits constructed
	from the universal set of gates shown in Figure~\ref{figCircuitGenerators}
	(the \CNOT, parameterised Pauli Z, and Hadamard gates \cite{NielsonChuang})
	are ZX diagrams.
	As shown in that figure each gate has a ZX-calculus analogue.
	Other common gates can easily be expressed in terms of these gates,
	for example $S := Z_{\frac\pi2}$ and $T := Z_{\frac\pi4}$ as well as $X_\alpha$ and CZ shown below:
	\begin{align}\label{eq:zx-derived-gates}
		X_\alpha & = \tikzfig{ZX/Xalpha} &
		\CZ & = \tikzfig{ZX/CZ}
	\end{align}

	\begin{figure}
		\centering
		\begin{tabular}{r c c c}
			Gate                       & \CNOT                        & $Z_\alpha$                   & H                                              \\[2em]
			\begin{tabular}{@{}r@{}}Circuit Diagram \\ (Inputs at the bottom)\end{tabular} & \tikzfig{circuit/CNOT}       & \node{box}{Z_\alpha}         & \node{box}{\text{H}}                           \\[2em]
			ZX Diagram                 & \tikzfig{ZX/CNOT}            & \node{gn}{\alpha}            & \node{h}{}                                     \\[2em]
			Matrix Interpretation      & $\begin{pmatrix}
					1 & 0 & 0 & 0 \\
					0 & 1 & 0 & 0 \\
					0 & 0 & 0 & 1 \\
					0 & 0 & 1 & 0 \\
				\end{pmatrix}$ & $\begin{pmatrix}
					1 & 0            \\
					0 & e^{i \alpha}
				\end{pmatrix}$ & $\frac{1}{\sqrt{2}}\begin{pmatrix}
					1 & 1  \\
					1 & -1
				\end{pmatrix}$
		\end{tabular}
		\caption{\label{figCircuitGenerators}A universal set of gates for quantum circuits,
			their ZX counterparts, and their interpretation as matrices acting on Hilbert space.}
	\end{figure}

\end{example}
\begin{remark}
	Note that there are wires in the depictions of the \CNOT\ and \CZ\ gates that are horizontal,
	and so it is ambiguous whether they are connected to inputs or outputs.
	This is a reflection of the `only connectivity matters' rule of ZX;
	any deformation of the diagram, provided it preserves the connectivity of the wires,
	results in another ZX diagram with the same interpretation.
	We can therefore draw horizontal wires without ambiguity.
\end{remark}

\section{Spiders and non-commutative groups} \label{secSpiders}

Spiders were introduced by Coecke and Duncan in the paper `Interacting Quantum Observables' \cite{Coecke11Observables},
and have already been exhibited in this paper as the red and green spiders of ZX.
The Observable Structures of that paper (also called spiders, Definition~6.4)
are commutative monoids over a given $\dagger$-SMC, along with other properties.
This commutativity was then vital to their Decorated Spider Rule~\cite[Theorem 7.11]{Coecke11Observables},
exhibited for ZX as the rule (S) of Figure~\ref{figZXRules}.
Our first result will be to show that \emph{any} monoid acting on $\Hilbert$ is commutative.

\begin{definition}[Monoid over $\Hilbert$]  \label{defMonoidOverQubits} In the manner of \cite[Definition~6.1]{Coecke11Observables}:
	A monoid over \Hilbert\ is a set $M$ of distinct states in \Hilbert\
	and an associative multiplication gate $\mu$.
	One of the states, $e$, is the unit for $\mu$.
	We depict $\mu$ and the elements of $M$ graphically as:
	\begin{align*}
		\binary[white]{} & : \bbH^{\tensor 2} \to \bbH &   & \set{\state{white}{m}  : \bbC \to \bbH }_{m \in M} \\
	\end{align*}
\end{definition}

\begin{restatable}{proposition}{propQubitModelsCommutative}\label{propQubitModelsCommutative}
	Every monoid over \Hilbert\ is commutative.
\end{restatable}

The proof of Proposition~\ref{propQubitModelsCommutative} is pure linear algebra and is found in \S\ref{prfPropQubitModelsCommutative}.
Our interest in this result is that only commutative
monoids can be modelled over $\Hilbert$ in a non-degenerate manner.
Since \UQuat\ is a non-commutative group,
and since the action of a group $\mu : G \times G \to G$ is also necessarily a monoid,
this means that we cannot faithfully model the group structure of $\UQuat$
as a monoid over  $\Hilbert$.

\begin{corollary} \label{corUQuatNotMonoid}
	There is no monoid over  $\Hilbert$ such that the monoid action $(M, \mu)$ is isomorphic
	to the group action of $\UQuat$.
\end{corollary}

By Corollary~\ref{corUQuatNotMonoid} we cannot find a spider
with phases labelled by unit quaternions that would obey the generalised spider
law \cite[Theorem~4]{Coecke11Observables}.

\section{The language ZQ} \label{secZQ}

Corollary~\ref{corUQuatNotMonoid} shows that we cannot simply
change ZX by labelling spiders with unit quaternions,
and so we have constructed a different approach.
We will use unit quaternion labels on directed edges to indicate rotations,
and use the phase-free Z spiders of ZX to mediate entanglement.
We present the graphical calculus ZQ as a compact closed PROP
generated by the morphisms in Figure~\ref{figZQGenProp}
and then present the interpretation of these generators in Figure~\ref{figZQInterpretation}.
We build the transpose of the $Q_q$ node in the usual way, as shown in Figure~\ref{figQTranspose}.

\begin{definition}[ZQ]  \label{defZQ}
	The graphical calculus ZQ is formed of:
	\begin{itemize}
		\item The generators of Figure~\ref{figZQGenProp}
		\item The interpretation of Figure~\ref{figZQInterpretation}
		\item The rules of Figure~\ref{figZQRules}
	\end{itemize}
\end{definition}

\begin{theorem}[ZQ is sound] \label{thmZQSound}
	The rules of ZQ are sound with respect to the standard interpretation.
\end{theorem}

\begin{proof}
	\label{prfThmZQSound}
	This proof is covered in \S \ref{secZQSound},
	since it amounts to just evaluating each side of each rule.
\end{proof}

\begin{theorem}[ZQ is complete] \label{thmZQComplete}
	The rules of ZQ are complete with respect to the standard interpretation.
\end{theorem}

\begin{proof}
	\label{prfThmZQComplete}
	This proof is covered in \S \ref{secZQComplete},
	and is performed by an equivalence with the ZX calculus, via the translation given in Figure~\ref{figZQZXTranslation}.
\end{proof}

\begin{theorem}[ZQ is universal] \label{thmZQUniversal}
	The diagrams of ZQ are universal for linear maps $\Hilbert^{\tensor m} \to \Hilbert^{\tensor n}$.
\end{theorem}

\begin{proof} \label{prfThmZQUniversal}
	The translation from ZX diagrams to ZQ diagrams exhibited in Figure~\ref{figZQZXTranslation}
	preserves interpretations (this is shown by inspection of the interpretations), and since ZX is universal therefore ZQ is universal.
\end{proof}

\begin{figure}[p]
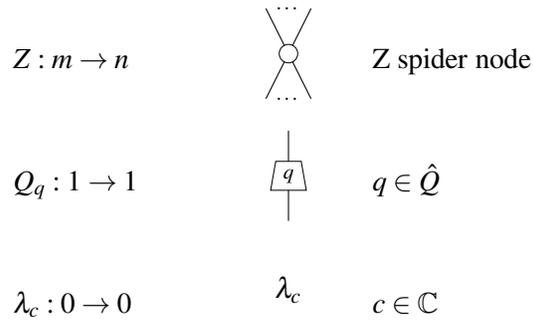

	\begin{center}
		\renewcommand{\arraystretch}{3}
		\tabcolsep=10pt
		\begin{tabular}[p]{m{1in} c m{2in}}
			$Z : m \to n$         & $\spider{smallZ}{}$ & Z spider node                \\
			$Q_{q} : 1 \to 1$     & $\node{qn}{q}$      & $q \in \UQuat$ \\
			$\lambda_c : 0 \to 0$ & $\lambda_c$         & $c \in \bbC$
		\end{tabular}
	\end{center}
	\caption{\label{figZQGenProp}The generators of ZQ as a PROP}
\end{figure}

\begin{figure}[p]
	\begin{align}
		\node{qnt}{q} & :=  \; \tikzfig{ZQ/q_y}
	\end{align}
	\caption{The transpose of $Q_q$ in ZQ\label{figQTranspose}}
\end{figure}

\begin{figure}[p]
	\begin{center}
		\begin{align}
			\interpret{\spider{smallZ}{}} & =
			\begin{pmatrix}
				1      & 0 & \dots  &   & 0      \\
				0      & 0 & \dots  &   & 0      \\
				\vdots &   & \ddots &   & \vdots \\
				0      &   & \dots  & 0 & 0      \\
				0      &   & \dots  & 0 & 1
			\end{pmatrix}
			  & \interpret{\dcup} & = \begin{pmatrix}
				1 \\ 0 \\ 0 \\ 1
			\end{pmatrix} \\[\rowgap]
			\interpret{\node{qn}{q}}      & =
			\begin{pmatrix}
				q_w - iq_z & -q_y + iq_x \\
				-q_y-iq_x  & q_w + iq_z
			\end{pmatrix} \label{eqnZQInterpretQ}&
			\interpret{\dcap}             & = \begin{pmatrix}
				1 & 0 & 0 & 1
			\end{pmatrix}\\[\rowgap]
			\interpret{\lambda_c}         & =
			c
		\end{align}
	\end{center}
	\caption{Interpretations of the generators of ZQ \label{figZQInterpretation}}
\end{figure}

\begin{figure}[p]
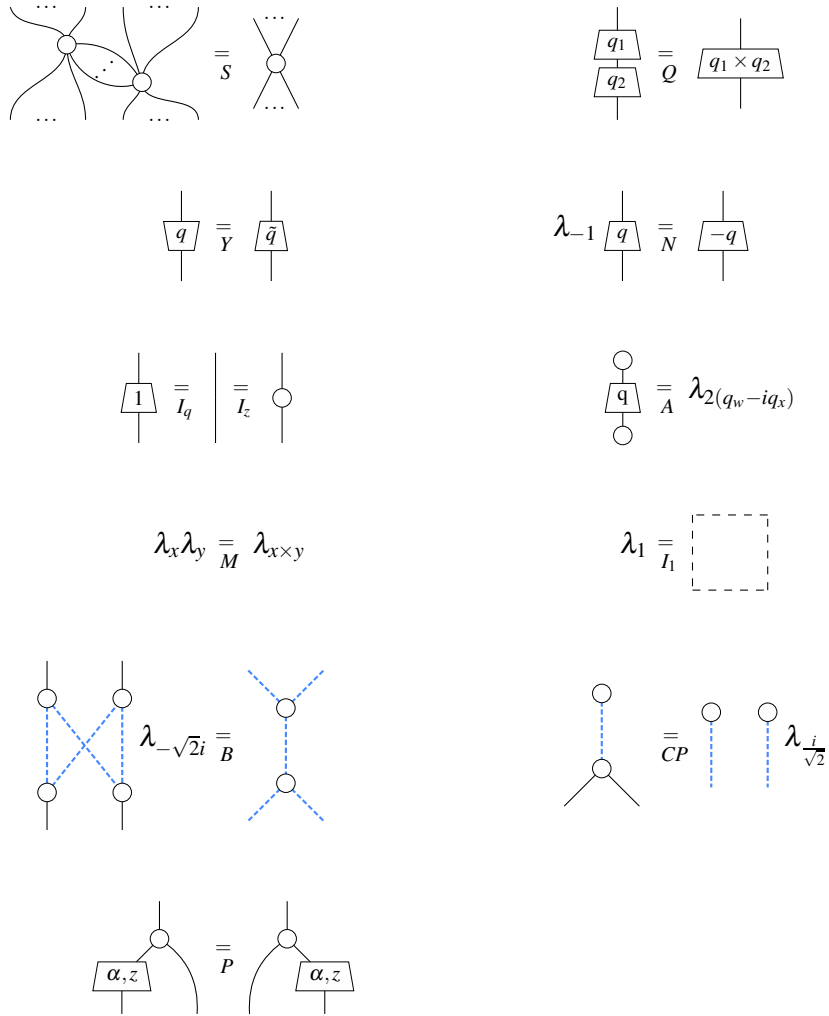

	\centering
	\begin{align*}
		{\tikzfig{ZQ/q_s_lhs}}                      & \by{S} \spider{smallZ}{}               &
		{\tikzfig{ZQ/q_lhs}} & \by{Q} \node{qn}{q_1 \times q_2} \\[2em]
		\node{qnt}{q}                               & \by{Y} \node{qn}{\tilde q}             &
		\lambda_{-1} \node{qn}{q} & \by{N} \node{qn}{-q}\\[2em]
		\node{qn}{1} \by{I_q}                       & \node{none}{} \by{I_z} \node{smallZ}{} &
		{\tikzfig{ZQ/q_l3_lhs}} & \by{A} \lambda_{2(q_w-iq_x)}\\[2em]
		\lambda_x \lambda_y                         & \by{M} \lambda_{x \times y}            &
		\lambda_1 & \by{I_1} \tikzfig{wire/empty} \\[2em]
		{\tikzfig{ZQ/q_b_lhs}} \lambda_{-\sqrt{2}i} & \by{B} {\tikzfig{ZQ/q_b_rhs}}          &
		{\tikzfig{ZQ/q_cp_lhs}} &\by{CP} {\tikzfig{ZQ/q_cp_rhs}}  \lambda_{\frac{i}{\sqrt{2}}} \\[2em]
		{\tikzfig{ZQ/q_phase_lhs}}                  & \by{P} {\tikzfig{ZQ/q_phase_rhs}}      &
	\end{align*}
	\caption{
		The rules of ZQ.\label{figZQRules}
		In rule S the diagonal dots indicate one or more wires, horizontal dots indicate zero or more wires.
		The right hand side of rule $I_1$ is the empty diagram,
		and $\tilde q$ is the quaternion $q$ reflected in the map $j \mapsto -j$.
	}
\end{figure}

\begin{remark} \label{remAngleVectorPair}
	Using angle-vector pair notation (Definition~\ref{defUQUatSO3}) we also have the interpretation:
	\begin{align}
		\interpret{\node{qn}{\alpha, v}} & =
		\begin{pmatrix}
			\cos \frac{\alpha}{2} - i \sin \frac{\alpha}{2} v_z & -i \sin\frac{\alpha}{2} (v_x + i v_y)                \\
			- i \sin\frac{\alpha}{2} (v_x - i v_y)              & \cos{\frac{\alpha}{2}} + i \sin\frac{\alpha}{2}{v_z}
		\end{pmatrix}  \label{eqnZQInterpretAlpha}
	\end{align}
	This is the same as the (transpose of the) operators $R_{v}(\alpha)$ or $e^{i\alpha (v \cdot \sigma)}$ \cite[\S{2.2}]{Wharton2015}.
	This transpose arises from the choice of $\pm Y$ as the canonical Pauli $Y$ matrix.
\end{remark}

\begin{definition}[Hadamard edge] \label{defHadamardEdge}
	In order to decrease diagrammatic clutter we shall use the following notation:
	\begin{align}
		\hedge & := \node{trap}{\pi, \frac{1}{\sqrt{2}}(x + z)} = \node{qn}{H}
	\end{align}
	This is a scaled version of the familiar `Hadamard edge' from e.g. \cite{2019arXiv190203178D},
	and we will use the shorthand $H$ rather than writing out $\pi, \frac{1}{\sqrt{2}}(x + z)$.
	Note that the Hadamard edge is symmetrical,
	but the $Q_H$ quaternion edge decoration is not,
	and so we will require a lemma to show that this Hadamard edge is well defined:
\end{definition}

\begin{lemma} \label{lemHadamardEdgeWellDefinedZQ}
	The Hadamard edge is well defined in ZQ,
	in that:
	\begin{align}
		ZQ \semantic \node{qn}{\text{H}} & = \hedge = \node{qnt}{H} & ZQ \entails \node{qn}{\text{H}} & = \hedge = \node{qnt}{H}
	\end{align}
\end{lemma}

\begin{proof} \label{prfLemHadamardEdgeWellDefinedZQ}
	For the semantics:
	\begin{align}
		\interpret{\node{qn}{\text{H}}} = \frac{-i}{\sqrt{2}}
		\begin{pmatrix}
			1 & 1  \\
			1 & -1
		\end{pmatrix} = \interpret{\node{qnt}{H}}
	\end{align}
	Syntactically:
	\begin{align}
		\node{qnt}{\cos \half[\pi] + \sin\half[\pi](i + k)}
		\by{Y}
		\node{qn}{\cos \half[\pi] + \sin\half[\pi](i + k)}
	\end{align}
\end{proof}

\subsection{Translation to and from ZX} \label{secZQZXTranslation}

We define the strict monoidal functors $F_X$ and $F_Q$ on generators in Figure~\ref{figZQZXTranslation}.
In defining this translation we make use of two facts:
Firstly that we can decompose any unit quaternion into
Z then X then Z rotations. This is tantamount to Euler Angle Decomposition
and is performed explicitly in Proposition~\ref{propZQQDecomp}.
Secondly we need to be able to express any complex number
in a rather particular form, which is shown in Lemma~\ref{lemRealNumberCanonicalForm}.

\begin{figure}[p]
	\begin{align}
		F_Q : &   & ZX \to                                                           & ZQ  \nonumber                                                                                                                                \\
		      &   & \spider{gn}{\alpha} \mapsto                                      & \tikzfig{ZQ/q_spider_expand} \lambda_{e^{i\alpha/2}}                                                                                         \\
		      &   & \spider{rn}{\alpha} \mapsto                                      & \left(\lambda_i \node{qn}{\text{H}}\right)^{\tensor n}
		\comp {\tikzfig{ZQ/q_spider_expand}} \lambda_{e^{i\alpha/2}}
		\comp \left(\lambda_i \node{qn}{\text{H}}\right)^{\tensor m} \\
		      &   & \node{h}{} \mapsto                                               & \node{qn}{\text{H}} \lambda_i                                                                                                                \\[\rowgap]
		F_X : &   & ZQ \to                                                           & ZX \nonumber                                                                                                                                 \\
		      &   & \node{qn}{q} \mapsto                                             & \ \tikzfig{ZX/abc} \galpharpi{-\alpha/2}\nonumber\galpharpi{-\beta/2}\galpharpi{-\gamma/2} \tripleblobs\tripleblobs\tripleblobs    \nonumber \\
		      &   & \spider{smallZ}{} \mapsto                                        & \spider{gn}{}                                                                                                                                \\
		      &   & \lambda_{\left(\sqrt{2}\right)^n e^{i\alpha} \cos \beta} \mapsto & \galpharpi{\alpha} \halfblobs \left(\rg{}{\pi}\right)^{\tensor n} \rg{\beta}{-\beta}
	\end{align}
	\caption{Translation from ZQ to ZX and back again.\label{figZQZXTranslation}
		The existence of $\alpha$, $\beta$ and $\gamma$ when translating the Q node is shown in Proposition~\ref{propZQQDecomp},
		likewise the decomposition of any complex number as $\left(\sqrt{2}\right)^n e^{i\alpha} \cos \beta$ is shown in Lemma~\ref{lemRealNumberCanonicalForm}.
	}
\end{figure}

\begin{restatable}{proposition}{propZQQDecomp}\label{propZQQDecomp}
	There exist $\alpha$ and $\gamma \in [0, 2\pi)$, and $\beta \in [0,\pi]$ such that:
	\begin{align}
		q_w + i q_x + j q_y + k q_z = \left(\cos \frac{\alpha}{2} + k \sin \frac{\alpha}{2}\right)\left(\cos \frac{\beta}{2} + i \sin \frac{\beta}{2}\right)\left(\cos \frac{\gamma}{2} + k \sin \frac{\gamma}{2}\right)
	\end{align}
\end{restatable}

The proof of this lemma is in \S\ref{prfPropZQQDecomp}.

\begin{restatable}{lemma}{lemRealNumberCanonicalForm}\label{lemRealNumberCanonicalForm}
	Any complex number $c$ can be expressed uniquely as $\left(\sqrt{2}\right)^n e^{i\alpha} \cos \beta$ where
	$n \in \bbN$, $\alpha \in [0,2\pi)$, $\beta \in [0, \pi)$
	and where $n$ is chosen to be the least $n$ such that $\sqrt{2}^n \geq \abs{c}$.
\end{restatable}

The proof of this lemma is in \S\ref{prfLemRealNumberCanonicalForm}.

\section{Conclusion}

This paper introduces the ZQ calculus,
showing it is sound, complete, and universal for qubit quantum computation.
What's more this paper has shown that simply extending the ZX-calculus
to allow arbitrary quaternions as phases would be fundamentally incompatible
with ZX's founding principle of spiders.
Despite this the completeness result for ZQ was shown via an equivalence with the ZX-calculus.

Additionally ZQ is, to the author's knowledge,
the first graphical calculus for qubits that uses a non-commutative phase group.
Indeed the only other qubit graphical calculus that uses a phase group that
is not a subgroup of $[0,2\pi)$
is the graphical calculus for Spekkens' Toy Bit Model in Ref.~\cite{Spekkens}.
This change in algebraic structure allows for the expression of the rules of ZQ
with a far simpler, although not eliminated, appeal to side conditions in comparison with the rules of Universal ZX.

With ZQ now described the author hopes that it will serve
as an Intermediate Representation for quantum circuit synthesis,
allowing it to benefit from the optimisation strategies of both TriQ \cite{TRIQ}
and ZX \cite{Backens2020circuit, 2019arXiv190203178D, kissinger2019reducing}.
The optimisation results of TriQ are not solely down to the use of quaternions
but also include routing and gate-decomposition concerns, which we have not addressed here.
Further practical work would include implementing such strategies,
and implementing ZQ in proof assistants such as Quantomatic \cite{Quantomatic} or PyZX \cite{PyZX}.
Further theoretical work would seek to eliminate the reliance on side conditions in the rules of ZQ,
and potentially adapt this calculus to express quaternionic quantum computing directly.

\FloatBarrier

\bibliographystyle{eptcs}
\bibliography{references}

\appendix

\section{Commutativity of Monoids over Hilbert Space proof}

\propQubitModelsCommutative*

\begin{proof} \label{prfPropQubitModelsCommutative}
	Taking a generic monoid $M$
	we look at the interpretation of the image of the generators of $\monprop_M$.
	For brevity we will just write $\interpret{D}$
	for $D$ a diagram in $\monprop_M$ to mean the interpretation
	of the image of $D$ in the model.
	We proceed by looking at the span of the interpretations of the elements of $M$.
	\begin{align}
		W & := span \set{\interpret{\state{white}{e}},\interpret{\state{white}{a}},\interpret{\state{white}{a'}}, \dots}
	\end{align}
	\begin{itemize}
		\item If $\dim W = 0$ then the monoid has only one element, $e$,
		      and so is commutative.
		\item If $\dim W > 0$ then
		      there either $M = \set{e}$ (and so commutative),
		      or there is some other element $a \in M$.
		      This implies that $\interpret{\state{white}{e}} \neq 0$, since:
		      \begin{align}
			      \text{assume} \quad \interpret{\state{white}{e}} & = 0                                                                       \\
			      \therefore \interpret{\tikzfig{monoid/md_0}}     & = 0 \quad \text{(any element $a$)}                                        \\
			      \interpret{\tikzfig{monoid/md_0}}                & = \interpret{\state{white}{a}}  \quad  \text{($e$ is the unit for $m$)}   \\
			      \therefore \interpret{\state{white}{a}}          & =\interpret{\state{white}{e}}  \contradiction \quad \text{$e$ and $a$ are distinct}
		      \end{align}
		\item If $\dim W = 1$ then without loss of generality:
		      \begin{align}\interpret{\state{white}{e}}                 & = \begin{pmatrix}
				      1 \\ 0
			      \end{pmatrix}                                                       \\
			      \therefore \interpret{\state{white}{a}}      & = \begin{pmatrix}
				      \lambda_a \\ 0
			      \end{pmatrix} \; \forall a\; \quad \text{ some } \lambda_a \in \bbC \\
			      \interpret{\tikzfig{monoid/md_0}}            & = \interpret{\state{white}{a}}  \quad  \text{$e$ is the unit for $m$}              \\
			      \therefore \interpret{\binary[white]{}}      & = \begin{pmatrix}
				      1 & \cdot & \cdot & \cdot \\
				      0 & \cdot & \cdot & \cdot
			      \end{pmatrix} \text{ where $\cdot$ represents unknowns}             \\
			      \therefore \interpret{\tikzfig{monoid/md_1}} & = \begin{pmatrix}
				      \lambda_a \lambda_b \\ 0
			      \end{pmatrix}= \interpret{\tikzfig{monoid/md_2}}
		      \end{align}
		\item if $\dim W = 2$ then the states span all of $\bbC^2$:
		      \begin{align}
			      \interpret{\state{white}{e}}                 & = \begin{pmatrix}
				      1 \\ 0
			      \end{pmatrix}        \quad \text{w.l.o.g.}                                                               \\
			      \tikzfig{monoid/md_0}                        & \by{unit_l} \state{white}{a} \quad \forall a                                                                            \\
			      \therefore \interpret{\tikzfig{monoid/m_2}}  & = \interpret{\node{none}{}} = \interpret{\tikzfig{monoid/m_1}} \quad \text{since } span \set{\state{white}{a}} = \bbC^2 \\
			      \therefore \interpret{\binary[white]{}}      & = \begin{pmatrix}
				      1 & 0 & 0 & \cdot \\
				      0 & 1 & 1 & \cdot
			      \end{pmatrix} \text{ where $\cdot$ represents unknowns}                                                  \\
			      \therefore \interpret{\tikzfig{monoid/m_3}}  & = \interpret{\binary[white]{}}                                                                                          \\
			      \therefore \interpret{\tikzfig{monoid/md_1}} & = \interpret{\tikzfig{monoid/md_2}}
		      \end{align}
	\end{itemize}
\end{proof}

\section{Quaternion / Rotation and Quaternion Decomposition Proofs}

\begin{proposition} \label{propZQPhi}
	The map $\phi$, given by
	\begin{align}
		  & \phi : & (\text{Unit Quaternions}, \times) & \to (2\times 2\text{ complex matrices}, \comp) \\
		  & \phi:  & q_w + iq_x + jq_y + kq_z          & \mapsto \begin{pmatrix}
			q_w - iq_z & q_y - iq_x \\
			-q_y-iq_x  & q_w + iq_z
		\end{pmatrix}
	\end{align}
	is a group homomorphism with trivial kernel.
\end{proposition}

\begin{proof} \label{prfPropZQPhi}

	Write $q_1$ as $w + x + y + z$ and $q_2$ as $w' + x' + y' + z'$:

	\begin{itemize}
		\item Show that $\phi(1) = \begin{pmatrix}
				      1 & 0 \\ 0 & 1
			      \end{pmatrix}$:
		      \begin{align}
			      \phi(1) = & \begin{pmatrix}
				      1 - i0 & 0 - i0 \\
				      -0 -i0 & 1 + i0
			      \end{pmatrix}
			      =         \begin{pmatrix}
				      1 & 0 \\ 0 & 1
			      \end{pmatrix}
		      \end{align}
		\item Show that $\phi(q_1)\phi(q_2) = \phi(q_1 \times q_2)$:
		      \begin{align}
			      LHS =         &
			      \begin{pmatrix}
				      w - iz & y - ix \\
				      -y-ix  & w + iz
			      \end{pmatrix} \comp
			      \begin{pmatrix}
				      w' - iz' & y' - ix' \\
				      -y'-ix'  & w' + iz'
			      \end{pmatrix} \\
			      LHS_{(1,1)} = & ( (ww' - zz'-yy' - xx') - i(wz' - xy' + yx' + w'z ) ) \\
			      LHS_{(1,2)} = & ( ((w -iz)(y' - i x') + (y-ix)(w'+iz'))               \\
			      =             & ( wy'- zx' + yw'+xz') - i  (wx' + xw' - yz' + zy')    \\
			      LHS_{(2,1)} = & ( (-y - ix)(w'-iz') + (w+iz)(-y'-ix') )               \\
			      =             & - (yw' + xz' + wy' - zx') - i (wx' + xw' - yz' + zy') \\
			      LHS_{(2,2)} = & ( (-y -ix)(y'-ix') + (w + iz)(w' + i z') )            \\
			      =             & ( ww'-xx'-yy'-zz') +i( wz' -xy' + yx' +zw')           \\
			      \nonumber\\
			      RHS_{(1,1)} = & (  ww' - xx'-yy'-zz') - i( wz' + xy'- yx' + zw')      \\
			      RHS_{(1,2)} = & ( wy' - xz' + yw') + zx' - i( wx' + xw' +yz' - zy')   \\
			      RHS_{(2,1)} = & -( wy' - xz' + yw') + zx')-i( wx' + xw' +yz' - zy')   \\
			      RHS_{(2,2)} = & ( ww' - xx'-yy'-zz') + i( wz' + xy'- yx' + zw')
		      \end{align}
		\item Show that $\phi(q) = \begin{pmatrix}
				      1 & 0 \\ 0 & 1
			      \end{pmatrix} \implies q = 1$:
		      Looking at the matrix entries individually:
		      \begin{align}
			      1 = & \ q_w - iq_z & \implies q_w & = 1 & \text{and}\quad  q_z & = 0 \\
			      0 = & -q_y - iq_x  & \implies q_y & = 0 & \text{and}\quad  q_x & = 0 \\
			      \therefore q = & 1
		      \end{align}
	\end{itemize}

\end{proof}

\propZQQDecomp*

\begin{proof} \label{prfPropZQQDecomp}
	\begin{align}
		RHS = & \left(\cos \frac{\alpha}{2} + k \sin \frac{\alpha}{2}\right)\left(\cos \frac{\beta}{2} + i \sin \frac{\beta}{2}\right)\left(\cos \frac{\gamma}{2} + k \sin \frac{\gamma}{2}\right) \\
		=     & \left(\cos \frac{\alpha}{2} \cos \frac{\beta}{2} \cos \frac{\gamma}{2} - \sin \frac{\alpha}{2} \cos \frac{\beta}{2} \sin \frac{\gamma}{2}\right) +                                 \\
		      & i \left(\cos \frac{\alpha}{2} \cos \frac{\beta}{2} \sin \frac{\gamma}{2} + \sin \frac{\alpha}{2} \cos \frac{\beta}{2} \cos \frac{\gamma}{2}\right) + \nonumber                     \\
		      & j \left(\cos \frac{\alpha}{2} \sin \frac{\beta}{2} \sin \frac{\gamma}{2} - \sin \frac{\alpha}{2} \sin \frac{\beta}{2} \cos \frac{\gamma}{2}\right) + \nonumber                     \\
		      & k \left(\cos \frac{\alpha}{2} \sin \frac{\beta}{2} \cos \frac{\gamma}{2} + \sin \frac{\alpha}{2} \sin \frac{\beta}{2} \sin \frac{\gamma}{2}\right) \nonumber                       \\
		=     & \cos \frac{\beta}{2} \left(\cos \frac{\alpha + \gamma}{2} + i \sin \frac{\alpha + \gamma}{2}\right) +
		j \sin \frac{\beta}{2} \left(\sin \frac{\gamma - \alpha}{2} -  i \cos \frac{\gamma - \alpha}{2}\right)\nonumber
	\end{align}
	From this we gather:
	\begin{align}
		q_w = & \cos \frac{\beta}{2} \cos \frac{\alpha + \gamma}{2} &
		q_x =& \cos \frac{\beta}{2} \sin \frac{\alpha + \gamma}{2} \\
		q_y = & \sin \frac{\beta}{2} \sin \frac{\gamma - \alpha}{2} &
		q_z =&\sin \frac{\beta}{2} \cos \frac{\gamma - \alpha}{2}
	\end{align}
	And finally use these to determine values of $\alpha$, $\beta$ and $\gamma:$
	\begin{itemize}
		\item $q_w^2 + q_x^2 = \cos^2 \frac{\beta}{2}$ determines up to two different possibilities of $\beta \in [0, 2\pi)$.
		      We will enforce $\beta \in [0, \pi]$ to make this unique and $\cos \frac{\beta}{2}$ non-negative.
		\item If $\beta = 0$ then set $\gamma = 0$, use $q_w$ and $q_x$ to determine $\alpha$
		\item Likewise if $\beta = \pi$ set $\gamma = 0$, use $q_y$ and $q_z$ to determine $\alpha$
		\item Otherwise determine $\alpha + \gamma / 2$ from $q_w$ and $q_x$,
		      and $\alpha - \gamma / 2$ from $q_y$ and $q_z$; their sum and difference give $2\alpha$ and $\gamma$ respectively.
	\end{itemize}

	The choices we made in this proof we justify by noting that we can represent these choices
	by certain applications of the spider rule (in the case $\beta = 0$) and $\pi$-commutativity rules
	(relating $(\alpha, \beta,\gamma) \sim (\alpha + \pi, -\beta, \gamma + \pi)$) in ZX.
\end{proof}

\lemRealNumberCanonicalForm*

\begin{proof} \label{prfLemRealNumberCanonicalForm}
	Express the complex number $c$ as $re^{i\alpha}$, where $r \in \bbR_{\geq 0}$.
	This matches our choice of $\alpha \in [0, 2\pi)$.
	For all $r$ there is at least one $n$ where $\sqrt{2}^n \geq r$ and so we can find a least such $n$.
	Once we know $n$ there is a unique $\beta \in [0,\pi)$ such that $\cos{\beta} \sqrt{2}^n = r$.
\end{proof}

\section{Soundness of ZQ} \label{secZQSound}

In this section we go through each of the rules given in Figure~\ref{figZQRules},
showing that the interpretations of the left and right hand sides of the rules are equal.

\begin{proposition} \label{propZQSSound}
	The rule S is sound:
	\begin{align}
		\interpret{{\tikzfig{ZQ/q_s_lhs_dots}}} & = \interpret{\spidermn{smallZ}{}{c+d}{a+b}}
	\end{align}
	Where there are $k \geq 1$ wires represented by \reflectbox{$\ddots$} in the middle of the left hand side.
\end{proposition}

\begin{proof} \label{prfPropZQSSound}
	This is simply a restating of the original Z spider law from \cite[Theorem~6.12]{Coecke08}.
\end{proof}

\begin{proposition} \label{propZQQSound}
	The rule $Q$ is sound:
	\begin{align}
		\interpret{{\tikzfig{ZQ/q_lhs}}} & = \interpret{\node{qn}{q_1 \times q_2}}
	\end{align}
\end{proposition}

\begin{proof}  \label{prfPropZQQSound}
	Follows from $\phi$ (see Definition~\ref{defQuaternions}) being a group isomorphism.
	The left hand side is multiplication in \SU2, the right hand side is multiplication in \UQuat.
\end{proof}

\begin{proposition} \label{propZQYSound}
	The rule $Y$ is sound:
	\begin{align}
		\interpret{\node{qn}{q_w + iq_x - jq_y + kq_z}} & = \interpret{\node{qnt}{q_w + iq_x + jq_y + kq_z}}
	\end{align}
\end{proposition}

\begin{proof} \label{prfPropZQYSound}
	The action of the cups and caps in Figure~\ref{figQTranspose}
	(where we defined the diagrammatic transpose),
	is to enact the transpose in the interpretation:
	\begin{align}
		\interpret{{\tikzfig{ZQ/q_y}}} =  \begin{pmatrix}
			q_w - iq_z & -q_y - iq_x \\
			q_y-iq_x   & q_w + iq_z
		\end{pmatrix}
		=                                 \interpret{\node{qn}{q_w + iq_x - jq_y + kq_z}}
	\end{align}
\end{proof}

\begin{proposition} \label{propZQNSound}
	The rule $N$ is sound:
	\begin{align}
		\interpret{\lambda_{-1} \node{qn}{q}}= \interpret{\node{qn}{-q}}
	\end{align}
\end{proposition}

\begin{proof} \label{prfPropZQNSound}
	\begin{align}
		LHS =  -1 \begin{pmatrix}
			q_w - iq_z & q_y - iq_x \\
			-q_y-iq_x  & q_w + iq_z
		\end{pmatrix}
		=      \begin{pmatrix}
			-q_w + iq_z & -q_y + iq_x \\
			q_y+iq_x    & -q_w - iq_z
		\end{pmatrix} = \, RHS
	\end{align}
\end{proof}

\begin{proposition} \label{propZQISound}
	The rules $I_q$ and $I_z$ are sound:
	\begin{align}
		\interpret{\node{qn}{1}} = \interpret{\node{none}{}} = \interpret{\node{smallZ}{}}
	\end{align}
\end{proposition}

\begin{proof} \label{prfPropZQISound}
	They all have the interpretation $\begin{pmatrix}
			1 & 0 \\ 0 & 1
		\end{pmatrix}$.
\end{proof}

\begin{proposition} \label{propZQASound}
	The rule $A$ is sound:
	\begin{align}
		\interpret{{\tikzfig{ZQ/q_l3_lhs}}} = \interpret{\lambda_{2(q_w-iq_x)}}
	\end{align}
\end{proposition}

\begin{proof} \label{prfPropZQASound}
	\begin{align}
		\interpret{{\tikzfig{ZQ/q_l3_lhs}}} = &
		\begin{pmatrix}
			1 & 1
		\end{pmatrix} \comp
		\begin{pmatrix}
			q_w - iq_z & q_y - iq_x \\
			-q_y-iq_x  & q_w + iq_z
		\end{pmatrix} \comp
		\begin{pmatrix}
			1 \\ 1
		\end{pmatrix} \\
		=                                     & q_w - iq_z + q_y - iq_x -q_y-iq_x + q_w + iq_z \\
		=                                     & 2(q_w - iq_x)                                  \\
		=                                     & \interpret{\lambda_{2(q_w - iq_x)}}
	\end{align}
\end{proof}

\begin{proposition} \label{propZQMSound}
	The rule $M$ is sound:
	\begin{align}
		\interpret{\lambda_x \lambda_y} = \interpret{\lambda_{x \times y}}
	\end{align}
\end{proposition}

\begin{proof} \label{prfPropZQMSound}
	Both sides have interpretation $x \times y$.
\end{proof}

\begin{proposition} \label{propZQL'Sound}
	The rule $I_\lambda$ is sound:
	\begin{align}
		\interpret{\lambda_1} = \interpret{\epsilon}
	\end{align}
	Where $\epsilon$ is the empty diagram.
\end{proposition}

\begin{proof} \label{prfPropZQL'Sound}
	Both sides have interpretation $1$.
\end{proof}

\begin{proposition} \label{propZQBSound}
	The rule $B$ is sound:
	\begin{align}
		\interpret{{\tikzfig{ZQ/q_b_lhs}} \lambda_{-\sqrt{2}i}} = \interpret{{\tikzfig{ZQ/q_b_rhs}}}
	\end{align}
\end{proposition}

\begin{proof} \label{prfPropZQBSound}
	\begin{align}
		LHS = & -\sqrt{2} i \times \begin{pmatrix}
			1 & 0 & 0 & 0 \\
			0 & 0 & 0 & 1
		\end{pmatrix}^{\tensor 2} \comp
		\left(\frac{-i}{\sqrt{2}}\begin{pmatrix}
			1 & 1 \\ 1 & -1
		\end{pmatrix}\right)^{\tensor 4} \comp \\
		      & \left(id_2 \tensor \begin{pmatrix}
				1 & 0 & 0 & 0 \\
				0 & 0 & 1 & 0 \\
				0 & 1 & 0 & 0 \\
				0 & 0 & 0 & 1
			\end{pmatrix} \tensor \id_2\right) \comp
		\begin{pmatrix}
			1 & 0 \\
			0 & 0 \\
			0 & 0 \\
			0 & 1
		\end{pmatrix}^{\tensor 2} \\
		=     & \frac{-i}{\sqrt{2}^3} \begin{pmatrix}
			1 & 1  & 1  & 1 \\
			1 & -1 & -1 & 1 \\
			1 & -1 & -1 & 1 \\
			1 & 1  & 1  & 1
		\end{pmatrix}                                 \\
		RHS = & \left(\frac{-i}{\sqrt{2}}\right)^5 \times \begin{pmatrix}
			1 & 1 \\ 1 & -1
		\end{pmatrix}^{\tensor 2}
		\comp
		\begin{pmatrix}
			1 & 0 \\
			0 & 0 \\
			0 & 0 \\
			0 & 1
		\end{pmatrix} \comp \\ &  \begin{pmatrix}
			1 & 1 \\ 1 & -1
		\end{pmatrix} \comp
		\begin{pmatrix}
			1 & 0 & 0 & 0 \\
			0 & 0 & 0 & 1
		\end{pmatrix} \comp  \begin{pmatrix}
			1 & 1 \\ 1 & -1
		\end{pmatrix}^{\tensor 2} \\
		=     & \left(\frac{-i}{\sqrt{2}}\right)^5 \times \begin{pmatrix}
			2 & 2  & 2  & 2 \\
			2 & -2 & -2 & 2 \\
			2 & -2 & -2 & 2 \\
			2 & 2  & 2  & 2
		\end{pmatrix}
		= \frac{-i}{\sqrt{2}^3} \times \begin{pmatrix}
			1 & 1  & 1  & 1 \\
			1 & -1 & -1 & 1 \\
			1 & -1 & -1 & 1 \\
			1 & 1  & 1  & 1
		\end{pmatrix}
	\end{align}
\end{proof}

\begin{proposition} \label{propZQCPSound}
	The rule $CP$ is sound:
	\begin{align}
		\interpret{{\tikzfig{ZQ/q_cp_lhs}}} = \interpret{{\tikzfig{ZQ/q_cp_rhs}}  \lambda_{\frac{i}{\sqrt{2}}}}
	\end{align}
\end{proposition}

\begin{proof} \label{prfPropZQCPSound}
	\begin{align}
		LHS = & \begin{pmatrix}
			1 & 1
		\end{pmatrix} \comp \frac{-i}{\sqrt{2}} \begin{pmatrix}
			1 & 1 \\ 1 & -1
		\end{pmatrix} \comp \begin{pmatrix}
			1 & 0 & 0 & 0 \\ 0 & 0 & 0 & 1
		\end{pmatrix}
		=  \frac{-i}{\sqrt{2}} \begin{pmatrix}
			2 & 0 & 0 & 0
		\end{pmatrix} \\
		RHS = & \frac{i}{\sqrt{2}} \left(\begin{pmatrix}
			1 & 1
		\end{pmatrix} \comp \frac{-i}{\sqrt{2}} \begin{pmatrix}
			1 & 1 \\ 1 & -1
		\end{pmatrix} \right)^{\tensor 2}
		=  \frac{-i}{\sqrt{2}} \begin{pmatrix}
			2 & 0 & 0 & 0
		\end{pmatrix}
	\end{align}
\end{proof}

\begin{proposition} \label{propZQPSound}
	The rule $P$ is sound:
	\begin{align}
		\interpret{{\tikzfig{ZQ/q_phase_lhs}}} = \interpret{{\tikzfig{ZQ/q_phase_rhs}}}
	\end{align}
\end{proposition}

\begin{proof} \label{prfPropZQPSound}
	\begin{align}
		LHS = & \begin{pmatrix}
			1 & 0 & 0 & 0 \\
			0 & 0 & 0 & 1
		\end{pmatrix} \comp \left(\begin{pmatrix}
				q_w - iq_z & 0           \\
				0          & q_w + i q_z
			\end{pmatrix} \tensor id_2 \right) \\
		=     & \begin{pmatrix}
			q_w - iq_z & 0 & 0 & 0           \\
			0          & 0 & 0 & q_w + i q_z
		\end{pmatrix}                                                              \\
		RHS = & \begin{pmatrix}
			1 & 0 & 0 & 0 \\
			0 & 0 & 0 & 1
		\end{pmatrix} \comp \left(id_2 \tensor \begin{pmatrix}
				q_w - iq_z & 0           \\
				0          & q_w + i q_z
			\end{pmatrix}\right)  \\
		=     & \begin{pmatrix}
			q_w - iq_z & 0 & 0 & 0           \\
			0          & 0 & 0 & q_w + i q_z
		\end{pmatrix}
	\end{align}
\end{proof}

\section{Completeness of ZQ} \label{secZQComplete}

The completion of ZQ is achieved by finding an equivalence between ZQ and ZX as PROPs.
We already know that ZX is complete \cite{UniversalComplete}
and this proof was by a similar equivalence with ZW, which was shown to be complete in Ref.~\cite{ZW}.
Equivalence is shown by finding a translation of the generators from ZX to ZQ and vice versa
(\S\ref{secZQZXTranslation}),
before then translating all of the rules from ZX into ZQ (\S\ref{secZXZQTranslatedRules}),
and keeping these as rules in ZQ.
Finally one has to ensure that any diagram translated from ZQ to ZX and back again
can be proven to be equivalent to the original ZQ diagram (\S\ref{secCompletenessZXZQBack}).
In symbols this is:
\begin{align}
	                                         & \interpret{D_1} = \interpret{D_2} \quad \text{Two diagrams in ZQ} \\
	ZX \entails                              & F_X D_1 = F_X D_2                                                 \\
	\S\ref{secZXZQTranslatedRules} \entails  & F_Q F_X D_1 = F_Q F_X D_1                                         \\
	\S\ref{secCompletenessZXZQBack} \entails & D_1 = F_Q F_X D_1 \qquad \text{and} \qquad D_2 = F_Q F_X D_2      \\
	\therefore ZQ \entails                   & D_1 = F_Q F_X D_1 = F_Q F_X D_2 = D_2
\end{align}

\subsection{Proving the translated ZX rules} \label{secZXZQTranslatedRules}

We aim to show that the rules translated from ZX are all derivable from the rules in Figure~\ref{figZQRules}, which we will refer to as $ZQ$.
We will use the ZX ruleset from \cite[Figure~2]{VilmartZX}, quoted here as Figure~\ref{figZXRules},
and refer to individual ZX rules as $ZX_{\text{rule name}}$.
To save space, we will assume applications of the $M$ rule (scalar multiplication) in the statements of the propositions.

\begin{lemma} \label{lemZQTranslationZ} Translation of the Z spider
	\begin{align}
		ZQ \entails F_Q\left(\spider{gn}{}\right) = \spider{smallZ}{}
	\end{align}
\end{lemma}

\begin{proof} \label{prfLemZQTranslationZ}
	\begin{align}
		LHS & = {\tikzfig{ZQ/q_spider_0}} \by{I_q} {\tikzfig{ZQ/q_spider_1}} \by{S} \spider{smallZ}{}
	\end{align}
\end{proof}

\begin{proposition} \label{propZQTranslationS} Translation of the Z spider rule
	\begin{align}
		ZQ     & \entails F_Q\left(ZX_S\right)                        \\
		\ie ZQ & \entails {\tikzfig{ZQ/q_ts1}} = {\tikzfig{ZQ/q_ts5}}
	\end{align}
	(The diagonal dots represent at least one wire between the Z spiders)
\end{proposition}

\begin{proof} \label{prfPropZQTranslationS}
	\begin{align}
		{\tikzfig{ZQ/q_ts1}}
		  & \by{S} {\tikzfig{ZQ/q_ts2}}    \\
		\by{P, Y} {\tikzfig{ZQ/q_ts3}}
		  & \by{Q, S} {\tikzfig{ZQ/q_ts4}} \\
		\by{P, Q, S} {\tikzfig{ZQ/q_ts5}}
	\end{align}
\end{proof}

\begin{proposition} \label{propZQTranslationIg} Translation of the Z spider identity
	\begin{align}
		ZQ     & \entails  F_Q\left(ZX_{I_g}\right)                                                \\
		\ie ZQ & \entails \node{smallZ}{} \comp \node{qn}{1} \comp \node{smallZ}{} = \node{none}{}
	\end{align}
\end{proposition}

\begin{proof} \label{prfPropZQTranslationIg}
	\begin{align}
		\node{smallZ}{} \comp \node{qn}{1} \comp \node{smallZ}{} \by{I_z} \node{qn}{1}
		  & \by{I_q} \node{none}{}
	\end{align}
\end{proof}

\begin{proposition} \label{propZQTranslationIr} Translation of the X spider identity
	\begin{align}
		ZQ     & \entails F_Q\left(ZX_{I_r}\right)                                                                                                                          \\
		\ie ZQ & \entails  \node{qn}{\text{H}} \comp \node{smallZ}{} \comp \node{qn}{1} \comp \node{smallZ}{} \comp \node{qn}{\text{H}} \lambda_i \lambda_i = \node{none}{}
	\end{align}
\end{proposition}

\begin{proof} \label{prfPropZQTranslationIr}
	\begin{align}
		LHS =    & \node{qn}{\text{H}} \comp \node{smallZ}{} \comp \node{qn}{1} \comp \node{smallZ}{} \comp \node{qn}{\text{H}} \lambda_i \lambda_i \\
		\by{I_z} & \node{qn}{\text{H}} \comp \node{qn}{1} \comp \node{qn}{\text{H}} \lambda_i \lambda_i                                             \\
		\by{I_q} & \node{qn}{\text{H}} \comp \node{qn}{\text{H}} \lambda_i \lambda_i
		\by{Q} \node{qn}{-1} \lambda_i \lambda_i \\
		\by{N}   & \node{qn}{1} \lambda_{-1} \lambda_i \lambda_i
		\by{I_q}  \node{none}{} \lambda_{-1} \lambda_i \lambda_i
		\by{M}  \node{none}{} \lambda_1
		\by{I_1}  \node{none}{}
	\end{align}
\end{proof}

We introduce our first three intermediate lemmas, corresponding to properties of the following three ZX diagrams:

\begin{align}
	\rg{\alpha}{\beta} \label{eqnSmallRG} , \qquad \tikzfig{ZX/HH}, \qquad \tikzfig{ZX/HHH}
\end{align}

\begin{lemma} \label{lemZQAlphaHBeta} Interaction of a Z state and Z effect joined by a Hadamard
	\begin{align}
		ZQ & \entails  {\tikzfig{ZQ/q_alpha_H_beta}} & = \lambda_{-\sqrt{2}\left(\left(\sin\frac{\alpha+\beta}{2}\right) + i \cos\frac{\alpha-\beta}{2}\right)}
	\end{align}
\end{lemma}

\begin{proof} \label{prfLemZQAlphaHBeta}
	\begin{align}
		{\tikzfig{ZQ/q_alpha_H_beta2}}
		              & \by{I_z}  {\tikzfig{ZQ/q_alpha_H_beta3}}
		\by{Q} {\tikzfig{ZQ/q_alpha_H_beta4}} \by{A} \lambda_{-\sqrt{2}((\sin\frac{\alpha+\beta}{2}) + i \cos\frac{\alpha-\beta}{2})} \\[\rowgap]
		\text{Since } & (\cos \frac{\alpha}{2} + k \sin \frac{\alpha}{2}) \times H \times (\cos \frac{\beta}{2} + k \sin \frac{\beta}{2}) =\nonumber \\
		              & \frac{1}{\sqrt{2}}(
		-\sin\frac{\alpha+\beta}{2} +
		i(\cos\frac{\alpha-\beta}{2}) +
		j(\sin\frac{\alpha-\beta}{2}) +
		k(\cos\frac{\alpha+\beta}{2}))
	\end{align}

\end{proof}

\begin{lemma} \label{propZQHH} Interaction of two Hadamard rotations
	\begin{align}
		ZQ \entails & \node{qn}{\text{H}} \comp \node{qn}{\text{H}} = \lambda_{-1} \node{none}{}
	\end{align}
\end{lemma}

\begin{proof} \label{prfPropZQHH}
	\begin{align}
		\node{qn}{\text{H}} \comp \node{qn}{\text{H}} \by{Q} \node{qn}{H \times H} = \node{qn}{-1} \by{N} \node{qn}{1} \lambda_{-1} \by{I_q} \node{none}{} \lambda{-1}
	\end{align}
\end{proof}

\begin{lemma} \label{propZQHHH} The value of the scalar describing three Hadamard rotations in parallel
	\begin{align}
		ZQ \entails {\tikzfig{ZQ/q_HHH}} & = \lambda_{\frac{i}{\sqrt{2}}}
	\end{align}
\end{lemma}

\begin{proof} \label{prfPropZQHHH}
	\begin{align}
		{\tikzfig{ZQ/q_HHH}}
		\by{S,\,\ref{propZQHH}} & \lambda_{-1} {\tikzfig{ZQ/q_HHH2}}          & \by{Y,S} \lambda_{-1}\ {\tikzfig{ZQ/q_HHH3}}                             \\
		\by{M, B}               & \lambda_{-i/\sqrt{2}}{\tikzfig{ZQ/q_HHH4}}  & \by{Y,S,I_z} \lambda_{-i/\sqrt{2}}\ {\tikzfig{ZQ/q_HHH5}}                \\
		\by{CP, Y}              & \lambda_{i/2\sqrt{2}}{\tikzfig{ZQ/q_HHH6}}  & \by{I_z}  \lambda_{i/2\sqrt{2}}\ {\tikzfig{ZQ/q_HHH7}}                   \\
		\by{M,\ref{propZQHH}}   & \lambda_{-i/2\sqrt{2}}{\tikzfig{ZQ/q_HHH8}} & \by{A}  \lambda_{-i/2\sqrt{2}}\ \lambda_{-i\sqrt{2}}\lambda_{-i\sqrt{2}} \\
		\by{M}& \lambda_{\frac{i}{\sqrt{2}}}
	\end{align}
\end{proof}

\begin{proposition} \label{propZQTranslationIV} Translation of the IV rule
	\begin{align}
		ZQ \entails     & F_Q\left(ZX_{IV}\right)          \\
		\ie ZQ \entails & \lambda_{e^{i \frac{\alpha}{2}}}
		{\tikzfig{ZQ/q_IV_1}}= \epsilon
	\end{align}
\end{proposition}

\begin{proof} \label{prfPropZQTranslationIV}
	\begin{align}
		\lambda_{e^{i \frac{\alpha}{2}}}
		{\tikzfig{ZQ/q_IV_1}} \by{I_z, I_q, S, M}   & \lambda_{e^{i \frac{\alpha}{2}}}  {\tikzfig{ZQ/q_IV_2}}                                             \\
		\by{\ref{lemZQAlphaHBeta}, \ref{propZQHHH}} & \lambda_{e^{i \frac{\alpha}{2}}} \lambda_{-\sqrt{2}ie^{-i \frac{\alpha}{2}}} \lambda_{i / \sqrt{2}} \\
		\by{M}                                      & \lambda_{1} \by{I_1} \epsilon
	\end{align}
\end{proof}

\begin{proposition} \label{propZQTranslationCP} Translation of the CP rule
	\begin{align}
		ZQ \entails     & F_Q\left(ZX_{CP}\right)                                               \\
		\ie ZQ \entails & \lambda_{-1} {\tikzfig{ZQ/q_CP1}} = {\tikzfig{ZQ/q_CP3}} \lambda_{-1}
	\end{align}
\end{proposition}

\begin{proof} \label{prfPropZQTranslationCP}
	\begin{align}
		LHS \by{1_z, 1_q} & \lambda_{-1} {\tikzfig{ZQ/q_CP2}} \by{A} \lambda_{-1} \lambda_{-i\sqrt{2}}{\tikzfig{ZQ/q_CP2a}}                              \\
		\by{CP}           & \lambda_{-1} \lambda_{-i\sqrt{2}} \lambda_{i / \sqrt{2}} {\tikzfig{ZQ/q_CP3}} \by{M} \lambda_{-1} {\tikzfig{ZQ/q_CP3}} = RHS
	\end{align}
\end{proof}

\begin{proposition} \label{propZQTranslationB} Translation of the B rule
	\begin{align}
		ZQ \entails     & F_Q\left(ZX_{B}\right)                                               \\
		\ie ZQ \entails & \lambda_{-i}\ {\tikzfig{ZQ/q_b1}} = {\tikzfig{ZQ/q_b4}} \lambda_{-i}
	\end{align}
\end{proposition}

\begin{proof} \label{prfPropZQTranslationB}
	\begin{align}
		LHS \by{I_z, I_q} & \lambda_{-i}{\tikzfig{ZQ/q_b2}} \by{B} & \lambda_{-i} \lambda_{i / \sqrt{2}} {\tikzfig{ZQ/q_b3}} \\
		\by{A, M, \ref{propZQHH}}&\lambda_{-i} \tikzfig{ZQ/q_b4}
	\end{align}
\end{proof}

\begin{proposition} \label{propZQTranslationH} Translation of the H rule
	\begin{align}
		ZQ \entails     & F_Q\left(H\right)                                                                                                                    \\
		\ie ZQ \entails &
		\left(\lambda_i \node{qn}{\text{H}}\right)^{\tensor n} \comp
		\left(\lambda_i \node{qn}{\text{H}}\right)^{\tensor n} \comp
		{\tikzfig{ZQ/q_spider_expand}}
		\lambda_{e^{i\alpha/2}} \comp\\
		                & \quad \left(\lambda_i \node{qn}{\text{H}}\right)^{\tensor m}  \comp \left(\lambda_i \node{qn}{\text{H}}\right)^{\tensor m} \nonumber \\
		                & = {\tikzfig{ZQ/q_spider_expand}} \lambda_{e^{i\alpha/2}}
	\end{align}
\end{proposition}

\begin{proof} \label{prfPropZQTranslationH}
	\begin{align}
		LHS & \by{\ref{propZQHH}} \left(\lambda_i \lambda_i \lambda_{-1}\right)^{\tensor n}  \left(\lambda_i \lambda_i \lambda_{-1}\right)^{\tensor m}
		{\tikzfig{ZQ/q_spider_expand}} \lambda_{e^{i\alpha/2}}
		\by{M} {\tikzfig{ZQ/q_spider_expand}} \lambda_{e^{i\alpha/2}}
	\end{align}
\end{proof}

Before proving the translation of the (EU') rule (Proposition~\ref{propZQTranslationEU}) we introduce some helpful lemmas.
ZQ $\entails F_Q\left(EU'\right)$.

\begin{lemma} \label{lemZQQLambdas}
	With the conditions of $ZX_{EU'}$
	\begin{align}
		\left(e^{i\left(\beta_1 + \beta_2 + \beta_3 + \gamma + 9\pi\right)/2}\right) \times
		\left(\frac{i}{\sqrt{2}}\right) \times
		\left(-\sqrt{2}\left(e^{i\gamma/2}\right)\right) = e^{i\left(\alpha_1 + \alpha_2 + \pi\right)/2}
	\end{align}
\end{lemma}

\begin{proof} \label{prfLemZQQLambdas}
	\begin{align}
		LHS = & \left(e^{i\left(\beta_1 + \beta_2 + \beta_3 + \gamma + 9\pi\right)/2}\right) \times
		\left(\frac{i}{\sqrt{2}}\right) \times
		\left(-\sqrt{2}\left(e^{i\gamma/2}\right)\right) \\
		=     & \left(-i\right)\left(e^{i\left(\beta_1 + \beta_2 + \beta_3 + 2\gamma + \pi\right)/2}\right)                                           \\
		=     & \left(-i\right)\left(e^{i\left(\arg z + \arg z' + \beta_2 + \arg z - \arg z' + 2x^+ - 2 \arg z + \pi - \beta_2 + \pi\right)/2}\right) \\
		=     & \left(-i\right)\left(e^{i\left(2x^++2\pi\right)/2}\right)                                                                             \\
		=     & e^{i\left(\alpha_1 + \alpha_2 + \pi\right)/2}
	\end{align}
\end{proof}

\begin{lemma}\label{lemZQHComm}
	The quaternion $\left(\pi, \frac{x+z}{\sqrt{2}}\right)$ and its interactions with $\left(\alpha, z\right)$ and $\left(\alpha, x\right)$:
	\begin{align}
		\left(\pi, \frac{x+z}{\sqrt{2}}\right) \times \left(\alpha,z\right) = \left(\alpha, x\right) \times \left(\pi, \frac{x+z}{\sqrt{2}}\right) \\
		\left(\alpha,z\right) \times \left(\pi, \frac{x+z}{\sqrt{2}}\right) = \left(\pi, \frac{x+z}{\sqrt{2}}\right) \times \left(\alpha,x\right)  \\
		\left(\pi, \frac{x+z}{\sqrt{2}}\right) \times \left(\pi, \frac{x+z}{\sqrt{2}}\right) = -1
	\end{align}
\end{lemma}

\begin{proof} \label{prfLemZQHComm}
	\begin{align}
		(\pi, \frac{x+z}{\sqrt{2}}) \times (\alpha, z) =                 & \frac{1}{\sqrt{2}}(i+k)(\cos \frac{\alpha}{2} + k \sin \frac{\alpha}{2})                                              \\
		=                                                                & \frac{1}{\sqrt{2}}(-\sin\frac{\alpha}{2} + i \cos \frac{\alpha}{2} - j \sin \frac{\alpha}{2} +k \cos \frac{\alpha}{2} \\
		=                                                                & \frac{1}{\sqrt{2}}(\cos\frac{\alpha}{2}+i\sin\frac{\alpha}{2})(i+k)                                                   \\
		=                                                                & (\alpha, x) \times (\pi, \frac{x+z}{\sqrt{2}})                                                                        \\
		\nonumber \\
		(\pi, \frac{x+z}{\sqrt{2}})\times(\alpha, x) =                   & \frac{1}{\sqrt{2}}(i+k)(\cos \frac{\alpha}{2} + i \sin \frac{\alpha}{2})                                              \\
		=                                                                & \frac{1}{\sqrt{2}}(-\sin\frac{\alpha}{2} + i \cos \frac{\alpha}{2} + j \sin \frac{\alpha}{2} +k \cos \frac{\alpha}{2} \\
		=                                                                & \frac{1}{\sqrt{2}}(\cos\frac{\alpha}{2}+k\sin\frac{\alpha}{2})(i+k)                                                   \\
		=                                                                & (\alpha, z) \times (\pi, \frac{x+z}{\sqrt{2}})                                                                        \\
		\nonumber \\
		(\pi, \frac{x+z}{\sqrt{2}}) \times (\pi, \frac{x+z}{\sqrt{2}}) = & \frac{1}{\sqrt{2}}(i+k)\frac{1}{\sqrt{2}}(i+k)                                                                        \\
		=                                                                & \frac{1}{2}(-1-j-1+j) = -1
	\end{align}
\end{proof}

We reproduce the side conditions for the rule $\ZX_{EU'}$ for reference here:

\begin{displayquote}[Figure~2, A Near-Minimal Axiomatisation of ZX-Calculus
		for Pure Qubit Quantum Mechanics \cite{VilmartZX}]
	In rule (EU'), $\beta_1$, $\beta_2$, $\beta_3$ and $\gamma$
	can be determined as follows: $x^+ := \frac{\alpha_1 + \alpha_2}{2}$,
	$x^- := x^-\alpha_2$, $z := - \sin (x^+) + i \cos(x^-)$
	and $z' := \cos(x^+) - i \sin (x^-)$,
	then $\beta_1 = \arg z + \arg z'$, $\beta_2 = 2 \arg(i + \frac{\abs{z}}{\abs{z'}})$,
	$\beta_3 = \arg z - \arg z'$, $\gamma = x^+ - \arg(z) + \frac{\pi - \beta_2}{2}$
	where by convention $\arg(0) := 0$ and $z' = 0 \implies \beta_2 = 0$.
\end{displayquote}

\begin{lemma} \label{lemZQEUQuaternion}
	With the conditions of $ZX_{EU'}$:
	\begin{align}
		\left(\alpha_1, z\right) \times H \times \left(\alpha_2, z\right) = H \times \left(\beta_1, z\right) \times H  \times \left(\beta_2,z\right) \times H \times \left(\beta_3, z\right) \times H
	\end{align}
\end{lemma}

In the hope of easing legibility we separate out the real, $i$, $j$, and $k$ components
of quaternions onto separate lines where suitable.

\begin{proof} \label{prfLemZQEUQuaternion}
	\begin{align}
		RHS =& H \times (\beta_1, z) \times H  \times (\beta_2,z) \times H \times (\beta_3, z) \times H \\
		= & H \times H \times H \times H \times (\beta_1, x) \times  (\beta_2,z) \times (\beta_3, x) & (\ref{lemZQHComm}) \\
		= & (\beta_1, x) \times  (\beta_2,z) \times (\beta_3, x)                                     & (\ref{lemZQHComm}) \\
		=& (\cos \beta_1 / 2 + i \sin \beta_1 / 2) \times  (\cos \beta_2 / 2 + k \sin \beta_2 / 2) \times  (\cos \beta_3 / 2 + i \sin \beta_3 / 2) \\
		=& 1(\cos \beta_1 /2 \cos \beta_2/2 \cos \beta_3/2 - \sin \beta_1 / 2 \cos \beta_2 /2 \sin \beta_3 / 2) +\\
		\nonumber & i (\cos \beta_1 /2 \cos \beta_2/2 \sin \beta_3/2 + \sin \beta_1 / 2 \cos \beta_2 /2 \cos \beta_3 / 2) + \\
		\nonumber&j (\cos \beta_1 /2 \sin \beta_2/2 \sin \beta_3/2 - \sin \beta_1 / 2 \sin \beta_2 /2 \cos \beta_3 / 2) + \\
		\nonumber&k (\cos \beta_1 /2 \sin \beta_2/2 \cos \beta_3/2 + \sin \beta_1 / 2 \sin \beta_2 /2 \sin \beta_3 / 2) \\
		= & 1 (\cos \beta_2/2)(\cos \beta_1 /2 \cos \beta_3/2 - \sin \beta_1 / 2 \sin \beta_3 / 2) + \\
		\nonumber & i (\cos \beta_2 /2)(\cos \beta_1 /2 \sin \beta_3/2 + \sin \beta_1 / 2  \cos \beta_3 / 2) + \\
		\nonumber & j (\sin \beta_2/2)(\cos \beta_1 /2  \sin \beta_3/2 - \sin \beta_1 / 2  \cos \beta_3 / 2) + \\
		\nonumber & k (\sin \beta_2/2)(\cos \beta_1 /2 \cos \beta_3/2 + \sin \beta_1 / 2\sin \beta_3 / 2) \\
		= & 1 (\cos \beta_2/2)(\cos \frac{\beta_1 + \beta_3}{2}) + \\
		\nonumber & i (\cos \beta_2 /2)(\sin \frac{\beta_1 + \beta_3}{2}) + \\
		\nonumber & j (\sin \beta_2/2)(\sin \frac{\beta_3 - \beta_1}{2}) + \\
		\nonumber & k (\sin \beta_2/2)(\cos \frac{\beta_1 - \beta_3}{2}) \\
		= & 1 (\cos \beta_2/2)(\cos \arg z) + \\
		\nonumber & i (\cos \beta_2 /2)(\sin \arg z) + \\
		\nonumber & j (\sin \beta_2/2)(- \sin \arg z') + \\
		\nonumber & k (\sin \beta_2/2)(\cos \arg z') \\
	\end{align}
	Using properties of arguments and moduli we then show the following:
	\begin{align}
		\cos \arg (a + ib) =     & a / \abs{a+ib}                                                                            \\
		\sin \arg (a + ib) =     & b / \abs{a+ib}                                                                            \\
		\abs{z}^2 =              & \sin(x^+)^2 + \cos(x^-)^2                                                                 \\
		\abs{z'}^2 =             & \sin(x^-)^2 + \cos(x^+)^2                                                                 \\
		\abs{z}^2 + \abs{z'}^2 = & \cos^2 x^+ + \sin^2 x^+ + \cos^2 x^- + \sin^2 x^- = 2                                     \\                                            \\
		\cos \arg z =            & \Re(z) / \abs{z} =  \frac{-\sin(\alpha_1 + \alpha_2)/2}{\abs{z}}                          \\
		\sin \arg z =            & \Im(z) / \abs{z} = \frac{\cos(\alpha_1 - \alpha_2)/2}{\abs{z}}                            \\
		\cos \arg z' =           & \Re(z') / \abs{z'}  = \frac{\cos(\alpha_1 + \alpha_2)/2}{\abs{z'}}                        \\
		\sin \arg z' =           & \Im(z') / \abs{z'} =\frac{-\sin(\alpha_1 - \alpha_2)/2}{\abs{z'}}                         \\
		\nonumber \\
		\cos(\beta_2 / 2) =      & \cos \arg (i + \abs{z}/\abs{z'})  = \cos \arg (\abs{z'}i + \abs{z})  = \abs{z} / \sqrt{2} \\
		\sin(\beta_2 / 2) =      & \sin \arg (i + \abs{z}/\abs{z'}) = \sin \arg (\abs{z'}i + \abs{z}) =\abs{z'} / \sqrt{2}
	\end{align}
	And now substitute these values into our expression for the right hand side:
	\begin{align}
		RHS =     & 1 (\cos \beta_2/2)(\cos \arg z) +                                           \\
		\nonumber & i (\cos \beta_2 /2)(\sin \arg z) +                                          \\
		\nonumber & j (\sin \beta_2/2)(- \sin \arg z') +                                        \\
		\nonumber & k (\sin \beta_2/2)(\cos \arg z')                                            \\
		=         & 1 (\abs{z} / \sqrt{2})(\frac{-\sin(\alpha_1 + \alpha_2)/2}{\abs{z}}) +      \\
		\nonumber & i (\abs{z} / \sqrt{2})(\frac{\cos(\alpha_1 - \alpha_2)/2}{\abs{z}}) +       \\
		\nonumber & j (\abs{z'} / \sqrt{2} )(- \frac{-\sin(\alpha_1 - \alpha_2)/2}{\abs{z'}}) + \\
		\nonumber & k (\abs{z'} / \sqrt{2} )(\frac{\cos(\alpha_1 + \alpha_2)/2}{\abs{z'}})      \\
		=         & (\frac{1}{\sqrt{2}}) \times                                                 \\
		\nonumber & (-1 (\sin(\alpha_1 + \alpha_2)/2) +                                         \\
		\nonumber & i (\cos(\alpha_1 - \alpha_2)/2) +                                           \\
		\nonumber & j (\sin(\alpha_1 - \alpha_2)/2) +                                           \\
		\nonumber & k (\cos(\alpha_1 + \alpha_2)/2))                                            \\
	\end{align}
	And now for the left hand side:
	\begin{align}
		LHS =     & (\alpha_1, z) \times H \times (\alpha_2, z)                                                                                             \\
		=         & \frac{1}{\sqrt{2}} (\cos\alpha_1 + k\sin\alpha_1) (i+k) (\cos\alpha_2 + k\sin\alpha_2)                                                  \\
		=         & \frac{1}{\sqrt{2}} (i \cos\alpha_1 \cos\alpha_2 - j \cos\alpha_1\sin\alpha_2 + k \cos\alpha_1 \cos\alpha_2 - \cos\alpha_1\sin\alpha_2 + \\
		\nonumber & j \sin\alpha_1\cos\alpha_2 + i \sin\alpha_1\sin\alpha_2  - \sin\alpha_1\cos\alpha_2 - k \sin\alpha_1\sin\alpha_2)                       \\
		=         & \frac{1}{\sqrt{2}} (-(\cos\alpha_1\sin\alpha_2 + \sin\alpha_1\cos\alpha_2)                                                              \\
		\nonumber & i(\sin\alpha_1\sin\alpha_2 + \cos\alpha_1\alpha_2) +                                                                                    \\
		\nonumber & j(\sin\alpha_1\cos\alpha_2 - \cos\alpha_1\sin\alpha_2) +                                                                                \\
		\nonumber & k(\cos\alpha_1\cos\alpha_2 - \sin\alpha_1\sin\alpha_2))                                                                                 \\
		=         & (\frac{1}{\sqrt{2}}) \times                                                                                                             \\
		\nonumber & (-1 (\sin(\alpha_1 + \alpha_2)/2) +                                                                                                     \\
		\nonumber & i (\cos(\alpha_1 - \alpha_2)/2) +                                                                                                       \\
		\nonumber & j (\sin(\alpha_1 - \alpha_2)/2) +                                                                                                       \\
		\nonumber & k (\cos(\alpha_1 + \alpha_2)/2))
	\end{align}
\end{proof}

\begin{lemma} \label{lemZQGamma}
	\begin{align}
		ZQ \entails {\tikzfig{ZQ/q_pigamma}} & = \lambda_{-\sqrt{2}\left(e^{i\gamma/2}\right)}
	\end{align}
\end{lemma}

\begin{proof} \label{prfLemZQGamma}
	This is a special case of Lemma~\ref{lemZQAlphaHBeta}
\end{proof}

\begin{proposition} \label{propZQTranslationEU}
	\begin{align}
		ZQ \entails     & F_Q\left(EU'\right)                                                                                                                                                         \\
		\ie ZQ \entails & \lambda_{e^{\frac{i}{2}\left(\alpha_1 + \alpha_2 + \pi\right)}} {\tikzfig{ZQ/q_EU1}} =\lambda_{e^{i(\beta_1 + \beta_2 + \beta_3 + \gamma + 9\pi)/2}} \ {\tikzfig{ZQ/q_eur}}
	\end{align}
\end{proposition}

\begin{proof} \label{prfPropZQTranslationEU}
	\begin{align}
		LHS \by{I_q, Q}              & \node{qn}{\left(\alpha_1, z\right) \times H \times \left(\alpha_2, z\right)} \lambda_{e^{i\left(\alpha_1 + \alpha_2 + \pi\right)/2}} \\
		RHS =                        &
		\overset{\lambda_{e^{i\left(\beta_1 + \beta_2 + \beta_3 + \gamma + 9\pi\right)/2}}}{{\tikzfig{ZQ/q_HHH}}}
		{\tikzfig{ZQ/q_pigamma}}
		\node{qn}{H \times \left(\beta_1, z\right) \times H  \times \left(\beta_2,z\right) \times H \times \left(\beta_3, z\right) \times H} \\
		\by{\ref{lemZQEUQuaternion}} &
		\lambda_{e^{i\left(\beta_1 + \beta_2 + \beta_3 + \gamma + 9\pi\right)/2}}
		{\tikzfig{ZQ/q_HHH}}
		{\tikzfig{ZQ/q_pigamma}}
		\node{qn}{\left(\alpha_1, z\right) \times H \times \left(\alpha_2, z\right)}  \\
		\by{\ref{lemZQGamma}}        &
		\lambda_{e^{i\left(\beta_1 + \beta_2 + \beta_3 + \gamma + 9\pi\right)/2}}
		{\tikzfig{ZQ/q_HHH}}
		\lambda_{-\sqrt{2}\left(e^{i\gamma/2}\right)}
		\node{qn}{\left(\alpha_1, z\right) \times H \times \left(\alpha_2, z\right)} \\
		\by{\ref{propZQHHH}}         &
		\lambda_{e^{i\left(\beta_1 + \beta_2 + \beta_3 + \gamma + 9\pi\right)/2}}
		\lambda_{\frac{i}{\sqrt{2}}}
		\lambda_{-\sqrt{2}\left(e^{i\gamma/2}\right)}
		\node{qn}{\left(\alpha_1, z\right) \times H \times \left(\alpha_2, z\right)} \\
		\by{\ref{lemZQQLambdas}}     &
		\lambda_{e^{i\left(\alpha_1 + \alpha_2 + \pi\right)/2}}
		\node{qn}{\left(\alpha_1, z\right) \times H \times \left(\alpha_2, z\right)}
	\end{align}
\end{proof}

We have shown that for every rule $L=R$ in ZX, $ZQ \entails F_Q\left(L\right) = F_Q\left(R\right)$.
We have therefore shown that if $ZX \entails D_1 = D_2$ then $ZQ \entails F_Q\left(D_1\right) = F_Q\left(D_2\right)$.

\subsection{From ZQ to ZX and back again} \label{secCompletenessZXZQBack}

It remains to be shown that $ZQ \entails F_Q\left(F_X\left(D\right)\right) = D$

\begin{proposition} \label{propZQRetranslateZSpider} Re-translating the Z spider
	\begin{align}
		ZQ \entails F_Q\left(F_X\left(\spider{smallZ}{}\right)\right) = \spider{smallZ}{}
	\end{align}
\end{proposition}

\begin{proof} \label{prfPropZQRetranslateZSpider}
	\begin{align}
		LHS = F_Q\left(\spider{gn}{}\right) \by{\ref{lemZQTranslationZ}} \spider{smallZ}{}
	\end{align}
\end{proof}

The following lemmas are necessary for the re-translation of the $Q$ node in Proposition~\ref{propZQRetranslateQ}.

\begin{lemma} \label{lemZQRetranslateTripleBlobs}
	\begin{align}
		ZQ \entails F_Q\left(\tripleblobs\right) = \lambda_{1 / \sqrt{2}}
	\end{align}
\end{lemma}

\begin{proof} \label{prfLemZQRetranslateTripleBlobs}
	\begin{align}
		F_Q\left(\tripleblobs\right) \by{M} \lambda_{-i} \tikzfig{ZQ/q_HHH} \by{M, \ref{propZQHHH}} \lambda_{1 / \sqrt{2}}
	\end{align}
\end{proof}

\begin{lemma} \label{propZQRetranslateGammaAlphaPi}
	\begin{align}
		ZQ \entails F_Q\left(\galpharpi{\gamma}\right) = \lambda_{-\sqrt{2}e^{i\gamma/2}}
	\end{align}
\end{lemma}

\begin{proof} \label{prfPropZQRetranslateGammaAlphaPi}
	\begin{align}
		F_Q\left(\galpharpi{\gamma}\right) \by{Y} \tikzfig{ZQ/q_pigamma} \lambda_{e^{i\gamma/2}} \lambda_{i} \lambda_{i} \by{\ref{lemZQGamma}, M} \lambda_{\sqrt{2}e^{i\gamma}}
	\end{align}

\end{proof}

\begin{lemma} \label{lemZQRetranslateZNode}
	\begin{align}
		ZQ \entails F_Q\left(\node{gn}{\alpha}\right) = \node{qn}{\alpha,z} \lambda_{e^{i\frac{\alpha}{2}}}
	\end{align}
\end{lemma}

\begin{proof} \label{prfLemZQRetranslateZNode}
	\begin{align}
		LHS = & \node{smallZ}{} \comp \node{qn}{\alpha,z} \comp \node{smallZ}{} \lambda_{e^{i\frac{\alpha}{2}}} =\node{qn}{\alpha,z}  \lambda_{e^{i\frac{\alpha}{2}}}
	\end{align}
\end{proof}

\begin{lemma} \label{lemZQRetranslateXNode}
	\begin{align}
		ZQ \entails F_Q\left(\node{rn}{\alpha}\right) = \node{qn}{\alpha,x} \lambda_{e^{i\frac{\alpha}{2}}}
	\end{align}
\end{lemma}

\begin{proof} \label{prfLemZQRetranslateXNode}
	\begin{align}
		LHS \by{M} & \node{qn}{\text{H}} \comp\node{smallZ}{} \comp  \node{qn}{\alpha,z} \comp\node{smallZ}{} \comp  \node{qn}{\text{H}} \lambda_{-e^{i\frac{\alpha}{2}}} \\
		\by{Q}     & \node{qn}{H \times \left(\alpha,z\right) \times H}  \lambda_{-e^{i\frac{\alpha}{2}}}
		\by{\ref{lemZQHComm}}
		\node{qn}{H \times H \times \left(\alpha,x\right)}  \lambda_{-e^{i\frac{\alpha}{2}}} \\
		           & \by{\ref{propZQHH}, N} \node{qn}{\alpha,x}  \lambda_{e^{i\frac{\alpha}{2}}}
	\end{align}
\end{proof}

\begin{proposition} \label{propZQRetranslateQ}
	\begin{align}
		ZQ \entails F_Q\left(F_X\left(\node{qn}{q}\right)\right) = \node{qn}{q}
	\end{align}
	Where $q = \left(\alpha_1, z\right) \times \left(\alpha_2, x\right) \times \left(\alpha_3 , z\right)$, as in Proposition~\ref{propZQQDecomp}
\end{proposition}

\begin{proof} \label{prfPropZQRetranslateQ}
	\begin{align}
		LHS =                                                                       & F_Q\left(
		\node{gn}{\alpha_1} \comp  \node{rn}{\alpha_2} \comp  \node{gn}{\alpha_3} \right) \tensor
		\galpharpi{-\alpha_1/2}
		\galpharpi{-\alpha_2/2}
		\galpharpi{-\alpha_3/2}
		\tripleblobs \tripleblobs \tripleblobs \\
		\by{\ref{propZQRetranslateGammaAlphaPi}, \ref{lemZQRetranslateTripleBlobs}} &
		F_Q\left(
		\node{gn}{\alpha_1} \comp  \node{rn}{\alpha_2} \comp  \node{gn}{\alpha_3}\right) \tensor
		\lambda_{\sqrt{2}e^{-i\alpha_1/2}}
		\lambda_{\sqrt{2}e^{-i\alpha_2/2}}
		\lambda_{\sqrt{2}e^{-i\alpha_3/2}}
		\lambda_{\frac{1}{\sqrt{2}}} \lambda_{\frac{1}{\sqrt{2}}} \lambda_{\frac{1}{\sqrt{2}}} \\
		\by{M}                                                                      & F_Q\left(
		\node{gn}{\alpha_1} \comp  \node{rn}{\alpha_2} \comp  \node{gn}{\alpha_3}\right) \tensor
		\lambda_{e^{-\frac{i}{2}\left(\alpha_1 + \alpha_2 + \alpha_3\right)}} \\
		\by{\ref{lemZQRetranslateXNode}, \ref{lemZQRetranslateZNode}}               &
		\left(\node{qn}{\alpha_1,z} \comp  \node{qn}{\alpha_2,x} \comp  \node{qn}{\alpha_3,z}\right) \tensor
		\lambda_{e^{-\frac{i}{2}\left(\alpha_1 + \alpha_2 + \alpha_3\right)}} \lambda_{e^{i\alpha_1/2}} \lambda_{e^{i\alpha_2/2}} \lambda_{e^{i\alpha_3/2}} \\
		\by{M}                                                                      & \node{qn}{\alpha_1,z} \comp  \node{qn}{\alpha_2,x} \comp  \node{qn}{\alpha_3,z}                      \\
		\by{Q}                                                                      & \node{qn}{\left(\alpha_1, z\right) \times \left(\alpha_2, x\right) \times \left(\alpha_3 , z\right)}
	\end{align}
\end{proof}

Finally we need the following lemma for Proposition~\ref{propZQRetranslateLambda}.

\begin{lemma} \label{propZQRetranslateRG}
	\begin{align}
		ZQ \entails F_Q\left(\rg{\beta}{-\beta}\right) = \lambda_{\sqrt{2}\cos \beta}
	\end{align}
\end{lemma}

\begin{proof} \label{prfPropZQRetranslateRG}
	\begin{align}
		LHS =  & {\tikzfig{ZQ/q_rgbb}} \lambda_{e^{i\beta_2}} \lambda_{e^{-i\beta_2}} \lambda_i
		\by{M} {\tikzfig{ZQ/q_rgbb}} \lambda_i
		\by{Q} {\tikzfig{ZQ/q_rgbb2}} \lambda_i \\
		\by{A} & \lambda_{\sqrt{2}\cos \beta}
	\end{align}
\end{proof}

\begin{proposition} \label{propZQRetranslateLambda}
	\begin{align}
		ZQ \entails F_Q\left(F_X\left(\lambda_{c}\right)\right) = \lambda_{c}
	\end{align}
\end{proposition}

\begin{proof} \label{prfPropZQRetranslateLambda}
	\begin{align}
		c =                            & \left(\sqrt{2}\right)^n e^{i\alpha} \cos \beta \quad \text{ for some } \alpha,\,\beta                \\
		LHS =                          & F_Q\left(\galpharpi{\alpha} \halfblobs \left(\rg{}{\pi}\right)^{\tensor n} \rg{\beta}{-\beta}\right)
		\by{\ref{propZQRetranslateGammaAlphaPi}, \ref{propZQHHH}} \lambda_{\sqrt{2}e^{i\alpha}} \lambda_{\frac{1}{\sqrt{2}}} \lambda_{\frac{1}{\sqrt{2}}} \left(\lambda_{\sqrt{2}}\right)^{\tensor n} F_Q\left(\rg{\beta}{-\beta}\right) \\
		\by{\ref{propZQRetranslateRG}} &
		\lambda_{\sqrt{2}e^{i\alpha/2}}
		\lambda_{\frac{1}{\sqrt{2}}}
		\lambda_{\frac{1}{\sqrt{2}}}
		\left(\lambda_{\sqrt{2}}\right)^{\tensor n}
		\lambda_{\sqrt{2}\cos\beta}
		\by{M}  \lambda_{\left(\sqrt{2}\right)^n e^{i\alpha} \cos \beta} = \lambda_c
	\end{align}
\end{proof}

We have shown that for each of the generators of ZQ, $ZQ \entails F_Q\left(F_X\left(G\right)\right) = G$,
and since $F_Q$ and $F_X$ are monoidal functors we know that $ZQ \entails F_Q\left(F_X\left(D\right)\right) = D$ for any diagram $D$.
This concludes our proof of completeness for the rules of ZQ.

\end{document}